\documentclass[11pt]{article}
\usepackage[T1]{fontenc}
\usepackage{lmodern}
\usepackage[margin=1in]{geometry}
\usepackage{amsmath,amssymb,amsthm,mathtools}
\usepackage{microtype,enumitem,booktabs}
\usepackage{tikz}
\usepackage[colorlinks=true,linkcolor=blue,citecolor=blue,urlcolor=blue]{hyperref}
\usetikzlibrary{arrows.meta,positioning,calc}
\numberwithin{equation}{section}
\newtheorem{theorem}{Theorem}[section]
\newtheorem{proposition}[theorem]{Proposition}
\newtheorem{corollary}[theorem]{Corollary}
\newtheorem{lemma}[theorem]{Lemma}
\theoremstyle{definition}

\newtheorem{remark}[theorem]{Remark}
\newcommand{\R}{\mathbb R}
\newcommand{\N}{\mathbb N}
\newcommand{\Z}{\mathbb Z}
\newcommand{\E}{\mathbb E}
\newcommand{\Prob}{\mathbb P}
\newcommand{\eps}{\varepsilon}
\newcommand{\A}{\mathcal A}
\newcommand{\1}{\mathbf 1}
\newcommand{\ceil}[1]{\left\lceil #1\right\rceil}
\newcommand{\floor}[1]{\left\lfloor #1\right\rfloor}
\DeclareMathOperator{\val}{val}
\DeclareMathOperator{\diam}{diam}
\DeclareMathOperator{\dist}{dist}
\DeclareMathOperator{\Unif}{Unif}

\title{Deterministic Causal Quantization with Bounded Error:\\ Adaptive Reconstruction and Entropy Bounds }
\author{Rama Cont\thanks{Mathematical Institute, University of Oxford.}
\and Massoud Heidari\thanks{LeastBit Information Lab, New York.}}
\date{\today}
\begin{document}

\maketitle
\begin{abstract}
We study deterministic causal quantization of real-valued data under a prescribed error bound: every sample must be reconstructed within tolerance $\varepsilon$ on every input sequence, using only present and past inputs. We construct an adaptive quantizer whose encoder and decoder maintain a common finite set of reconstruction values, updated from the transmitted symbols alone. The quantizer requires no knowledge of the source distribution, and suitable initialization reproduces any uniform grid. 

For stationary sources that can be quantized with finite entropy,
the scheme uses at most two additional bits per sample compared
with the best block quantizer satisfying the same error bound,
provided its initial reconstruction values are at least
$\varepsilon$ apart. This guarantee requires no density,
independence, or mixing assumptions. Initializing the scheme
to reproduce a uniform grid reduces the gap to one bit per sample.
  Rates are measured by the entropy of the quantizer output, which block
  entropy coding can approach.
Under entropy and Sobolev regularity conditions, the grid's excess
rate is at most $2\mathcal J\varepsilon^2/(\pi^2\ln 2)$, where
$\mathcal J$ is the asymptotic Fisher information per sample.
Thus the cost of causal reconstruction vanishes at a 
quadratic rate as tolerance decreases.
For sources uniform within separated intervals of width $2\varepsilon$, storing the interval centers attains the optimum and its Shannon lower bound.  These configurations can outperform every uniform grid of the same
cell width. Over sources with well-separated intervals, the sharp
worst-case excess rate of the best uniform grid is
$1/(2\ln 2)\approx0.72$ bits per sample.
\end{abstract}
\newpage
\tableofcontents
\newpage
\section{Introduction}\label{sec.introduction}

\subsection{Causal compression with deterministic error bound}

Many data streams are compressed under a per-sample error guarantee: every
reconstructed value must lie within a prescribed tolerance of the original.
This is the criterion of near-lossless image coding~\cite{Weinberger2000} and of
error-bounded compression of scientific data~\cite{DiCappello2016}. It certifies
each individual recording rather than an average over a source model, so it
can be checked without knowing the source distribution. For a real-valued
stream $(x_t)_{t\geq1}$ and a tolerance $\eps>0$, we require reconstructions
$\widehat x_t$ with
\begin{equation}\label{eq.pathwise-error}
 |x_t-\widehat x_t|\leq\eps\qquad(t\geq1)
\end{equation}
on every input sequence. We focus on causal schemes, in which $\widehat x_t$
depends only on $x_1,\dots,x_t$, so that each sample can be encoded as it
arrives; the transmitted symbols may then be entropy coded in blocks.

The benchmark is the least rate of any block code
meeting~\eqref{eq.pathwise-error}. For a stationary source $X$, the process
$\eps$-entropy $R_\eps(X)$ is the smallest entropy per sample of a partition of
$\R^n$ into cells contained in sup-norm balls of radius $\eps$, in the limit
$n\to\infty$~\cite{Posner1967,posner1971,Bucklew1991}. Block codes may look
ahead; causal codes may not. The causal floor $R^{c}_\eps(X)$ is the least rate
of a code whose reconstructions depend only on current and past inputs, with
block entropy coding allowed~\cite{NeuhoffGilbert1982}, so that
$R_\eps\leq R^{c}_\eps$. We focus here on the following questions: 
\begin{itemize}
\item Can an adaptive causal quantizer guarantee the error
bound~\eqref{eq.pathwise-error} on every input sequence, without
knowledge of the source distribution or side information
to synchronize encoder and decoder?
\item How close can its compression rate come to the optimum
$R_\eps(X)$ achievable by unrestricted block quantization?
\item When can stored reconstruction values outperform every
uniform grid with the same error tolerance?
\end{itemize}
We give explicit responses to these questions, for every tolerance level $\varepsilon$: causality costs little
under~\eqref{eq.pathwise-error}, the simplest causal scheme is already within
one bit per sample of the optimal rate $R_\eps(X)$, and reconstruction values matched to the
source can do better than any uniform grid. 
\subsection{Contributions}
The paper makes four main contributions.

\paragraph{An adaptive scheme with a distribution-free guarantee.}
We construct a deterministic causal quantizer satisfying
\eqref{eq.pathwise-error}, in which encoder and decoder maintain
the same set of $\omega$ reconstruction values and update it from
the transmitted symbols alone, without side information or knowledge
of the source distribution. A sample within $\eps$ of a stored value
is reconstructed by that value. Otherwise, the gap containing the
sample is divided into cells of width at most $2\eps$, anchored at
the neighboring stored values; the sample is reconstructed at the
midpoint of its cell, that midpoint is stored, and the oldest stored
value is discarded. We call this scheme the Midpoint Representative
Realization (MRR).

If the initial values are at least $\eps$ apart, the rate of MRR
is within two bits per sample of the block-code optimum $R_\eps(X)$
for every stationary source admitting finite-entropy scalar
quantization, without density, independence, or mixing assumptions
(Corollary~\ref{cor.two-bit-process}). The proof shows that every
interval of length $2\eps$ meets at most four cells of the adaptive
partition, and this overlap bound is sharp
(Theorem~\ref{thm.four-overlap}). It then counts the branches of the
tree of possible symbol sequences. More generally, if at every
reachable state each interval of length $2\eps$ meets at most $m$
cells, the same argument gives a $\log_2 m$-bit comparison for
synchronized causal quantizers under the hypotheses of
Proposition~\ref{prop.causal-overlap}. We are not aware of a comparable
distribution-free comparison for adaptive quantizers.

Let $\mathcal J=\limsup_{n\to\infty}J(f_n)/n$ be the asymptotic
Fisher information per sample of the block densities $f_n$, defined
in~\eqref{eq.block-fisher-definition} and~\eqref{eq.fisher-rate}.
Under the entropy and Sobolev regularity assumptions of
Theorem~\ref{thm.near-floor}, with $\mathcal J<\infty$, we also
obtain a bound from an arbitrary initial state:
\begin{equation}\label{eq.intro-gap}
 0\leq R_{\mathrm{MRR}}-R_\eps
 \leq \frac{2\mathcal J\eps^2}{\pi^2\ln2}
 +\frac{\eps\sqrt{\mathcal J}}{e\ln2}
   \left(\floor{\frac{3\omega}{2}}-1\right),
\end{equation}
where $R_{\mathrm{MRR}}$ denotes the entropy rate of the transmitted
symbols, as defined in the paper. The second term grows at most
linearly with the capacity $\omega$.

\paragraph{A uniform grid is within one bit per sample of the optimal rate.}
With consistent boundary ownership, initialization on a uniform grid
of spacing $2\eps$ makes MRR reproduce that grid
(Proposition~\ref{prop.lattice}); its transmitted symbols and the
grid indices have identical block entropy. This specialization gives
sharper guarantees, uniformly over capacity and grid origin.
With entropy coding of its indices, the uniform quantizer's rate
exceeds $R_\eps(X)$ by at most one bit per sample for every stationary
source admitting finite-entropy scalar quantization. No density,
independence, or mixing assumption is needed, and the constant is
sharp (Corollary~\ref{cor.lattice-one-bit}).

\paragraph{The vanishing cost of causal reconstruction.}
Under the entropy and Sobolev regularity assumptions of
Theorem~\ref{thm.near-floor}, grid initialization removes the
capacity term in~\eqref{eq.intro-gap}. The gap of the uniform grid
is therefore at most
$\min\{1,\,2\mathcal J\eps^2/(\pi^2\ln2)\}$ bits per sample;
see~\eqref{eq.lattice-combined-gap} and
Proposition~\ref{prop.lattice}. Unbounded scores, including Gaussian
ones, are allowed. Since the grid is itself causal, the price of
causal reconstruction under~\eqref{eq.pathwise-error} is
$O(\eps^2)$ as $\eps\downarrow0$
(Remark~\ref{rem.price-causality}). Here causality constrains
reconstruction, while subsequent block entropy coding is allowed.
This contrasts with the classical high-resolution space-filling
loss of $\tfrac12\log_2(\pi e/6)\approx0.254$ bits per sample for
entropy-coded scalar quantization under mean-squared error and
the corresponding regularity assumptions~\cite{LinderZamir2006}.
The geometric distinction is that the error ball under the
blockwise maximum norm is a cube, which tiles space by
translation~\cite{PosnerRodemich1969}.

For memoryless sources, the optimal causal rate equals the scalar
$\eps$-entropy, without the time-sharing that can be required under
average distortion~\cite{NeuhoffGilbert1982}
(Proposition~\ref{prop.causal-iid}). For a single sample satisfying
the corresponding entropy and Sobolev assumptions, the estimate
gives a nonasymptotic form of R\'enyi's
theorem~\cite{Renyi1959}: the entropy of a sample quantized with
step $\Delta$ exceeds $h(X)-\log_2\Delta$ by at most
$\Delta^2J(X)/(2\pi^2\ln2)$, for every $\Delta>0$ and every grid
offset (Lemma~\ref{lem.rectangle-kl}).

\paragraph{Exact optimal rates and source-matched reconstruction values.}
For sources with structure at the scale of the tolerance, the optimal rate
can be computed exactly and compared with every grid.
Let $X_t=c_{J_t}+U_t$, where the finitely many centers $c_j$ are
at least $2\eps$ apart, $J$ is stationary, and $(U_t)$ is i.i.d.,
uniform on $[-\eps,\eps)$, and independent of $J$.
Then the $\eps$-entropy equals its Shannon lower bound at every
block length. MRR initialized with exactly these centers attains
the optimum, whether or not they lie on a common grid
(Theorem~\ref{thm.exact-separated}). This includes histograms
with bin width $2\eps$, with initialization matched to the bin
centers.

Source-matched reconstruction values can outperform every uniform
grid. Over this source class with centers more than $4\eps$ apart,
the sharp worst-case excess rate of the best uniform grid of width
$2\eps$ is $1/(2\ln2)\approx0.72$ bits per sample (Corollary \ref{cor:sharp-uniform-band-penalty}). For two equally
weighted intervals with centers $5\eps$ apart, the excess is
exactly half a bit (Proposition~\ref{prop.capacity-advantage} and
Section~\ref{sec.exact-floor}). These comparisons use stored values
matched to the source, which MRR then keeps unchanged; whether
first-in, first-out replacement finds such values from an arbitrary
start is left open.

\subsection{Related work}

Adaptive quantization with synchronized encoder and decoder goes back to
Jayant's one-word-memory quantizer~\cite{Jayant1973}; see also Ortega and
Vetterli~\cite{ortega1997} and the survey~\cite{GrayNeuhoff1998}. Zero-delay
schemes that compete with the best scalar quantizer in hindsight have been
developed for individual sequences under average
distortion~\cite{LinderLugosi2001,GyorgyLinderLugosi2008}.
Bhatt et al.~\cite{BhattNazerOrdentlichPolyanskiy2021} bound
the number of quantization levels needed to preserve information
about a target variable; here we minimize output entropy under
a deterministic reconstruction-error constraint.
Kostina and Hassibi~\cite{KostinaHassibi2019} study causal compression with a
shared predictor and quadratic cost; here the shared state is a set of
reconstruction values and fidelity is~\eqref{eq.pathwise-error}.

Causal source codes were introduced in~\cite{NeuhoffGilbert1982}, and their
high-resolution loss under squared error was identified
in~\cite{LinderZamir2006}. Under expected difference distortion, the Shannon
lower bound is asymptotically tight~\cite{LinderZamir1994,Koch2016}, and
entropy-coded uniform quantization is asymptotically optimal among scalar
quantizers~\cite{gish1968}. These results are asymptotic; ours hold at every
tolerance for the pathwise constraint~\eqref{eq.pathwise-error}.
Nonasymptotic lattice-entropy bounds under density regularity, including for
maximum-coordinate distortion, are given in~\cite[Theorems~8 and~13]{Kostina2017};
for a fixed lattice, the decomposition in
Proposition~\ref{prop.rectangle-identity} reduces
to~\cite[Eq.~(88)]{Kostina2017}. Under the common assumption
$|\partial_i\ln f_n|\leq L$, the cubic-grid bound of~\cite{Kostina2017} gives an
entropy excess above the Shannon lower bound of at most $2L\eps/\ln2$ per
coordinate, whereas Lemma~\ref{lem.rectangle-kl}, which rests on the sharp
interval logarithmic Sobolev inequality~\cite{GhangMartinWaruhiu2014}, gives
$2L^2\eps^2/(\pi^2\ln2)$. Related entropy-continuity bounds via Wasserstein
distance appear in~\cite{PolyanskiyWu2016}. Asymptotic tightness of uniform
quantization is classical~\cite{Renyi1959,gish1968}; our grid estimates are
nonasymptotic. The four-cell bound is combined with the partition-overlap
inequality~\cite{Hochman2014}. The finite-alphabet computation in
Appendix~\ref{app.computation} specializes graph-entropy
optimization~\cite{Csiszar1990,Harangi2023} to the absolute-error constraint.

\subsection{Organization and notation}

Section~\ref{sec.ingredients} defines causal quantizers with a shared state, and
Section~\ref{sec.dynamics} constructs MRR and proves its pathwise and
synchronization properties; for i.i.d.\ inputs with a bounded, atomless distribution,
its age-ranked state is a Feller Markov chain
(Theorem~\ref{thm.feller}). Section~\ref{sec.overlap} proves
the four-cell bound and the scalar comparisons, including the two-bit bound for
single samples (Theorem~\ref{thm.two-bit}). Section~\ref{sec.process} treats
rates: the distribution-free process bound, the Fisher-information estimate,
uniform grids and the price of causality (Sections~\ref{sec.grid-certificates}
and~\ref{sec.fixed-grid}), and exact floors (Section~\ref{sec.exact-floor}).
It also gives certified brackets for autoregressive and independent Laplace and
Gaussian sources (Corollaries~\ref{cor.ar-laplace}, \ref{cor.iid-laplace-grid}
and~\ref{cor.ar-gaussian}), the capacity-one case
(Corollary~\ref{cor.capacity-one-grid}), histograms
(Corollary~\ref{cor.uniform-cell-shannon-equality},
Remark~\ref{rem.histograms}), the uniqueness of the two-interval configuration
(Remark~\ref{rem.isolated-geometry}), and a rate formula for independent inputs
in terms of limiting state laws (Lemma~\ref{lem.iid-rate},
Proposition~\ref{prop.invariant-state-rate}). Section~\ref{sec.discussion}
discusses limitations and open questions. Readers interested mainly in uniform
grids can go directly to Sections~\ref{sec.grid-certificates}
and~\ref{sec.fixed-grid}.

Throughout, $H$ and $h$ denote discrete and differential entropy in bits, and
$D_2$ is relative entropy with logarithms to base two; natural logarithms are
written $\ln$. Standard entropy identities and noiseless coding results are used
as in~\cite{CoverThomas2006}.

\section{Causal quantizers with a shared state}\label{sec.ingredients}
Our aim is to compress a real-valued stream sample by sample, with every
reconstruction within $\eps$ of the original, at a rate as close as possible to
the block optimum $R_\eps(X)$. A fixed quantizer, such as a uniform grid, meets
the error bound but cannot adapt its cells to the data. An adaptive quantizer
can, but the decoder must then know which cells the encoder is using at each
time. As in backward-adaptive quantization~\cite{Jayant1973,ortega1997}, we
achieve this without side information by letting encoder and decoder maintain a
common \emph{state}, which both sides update from the transmitted symbols alone.
This section defines this class of schemes, their error guarantee and their
rate; Section~\ref{sec.dynamics} constructs the scheme we analyze.

At each time, the encoder sees the new sample and the current state and
transmits a symbol from a countable alphabet $\mathcal S$; the decoder
reconstructs the sample from the state and the symbol; and both sides update
the state from the same information. Formally, a \emph{causal quantizer with a
shared state} consists of
a measurable state space $\mathcal A$, a countable alphabet $\mathcal S$, and
three Borel maps
\[
 S:\mathcal A\times\R\to\mathcal S,\qquad
 d:\mathcal A\times\mathcal S\to\R,\qquad
 V:\mathcal A\times\mathcal S\to\mathcal A,
\]
the encoder, the decoder and the state update. Starting from a common
deterministic state $A_0$, the scheme runs as
\begin{equation}\label{eq.shared-state}
 S_{t+1}=S(A_t,x_{t+1}),\qquad
 \widehat x_{t+1}=d(A_t,S_{t+1}),\qquad
 A_{t+1}=V(A_t,S_{t+1})\qquad(t\geq0).
\end{equation}
Because the state update uses only the transmitted symbol, the decoder
reproduces the encoder's state at every step and needs no side information. We
call this property \emph{decoder compatibility}; it is the only structural
requirement of the class. Writing $U(A,x)=V(A,S(A,x))$ for the induced update,
the state evolves as $A_{t+1}=U(A_t,x_{t+1})$.

For a fixed state $A$, the sets $C_A(s)=\{x:S(A,x)=s\}$ partition $\R$ into
cells, and the cell containing $x$ determines both the transmitted symbol and
the reconstruction $\pi_A(x)=d(A,S(A,x))$. The quantizer is
\emph{$\eps$-faithful} if
\begin{equation}\label{eq.faithful}
 |x-\pi_A(x)|\leq\eps\qquad\text{for every state }A\text{ and every }x\in\R,
\end{equation}
that is, if every cell lies in the closed interval of radius $\eps$ around its
reconstruction. An $\eps$-faithful quantizer satisfies~\eqref{eq.pathwise-error}
on every input sequence, and it is causal in the sense
of~\cite{NeuhoffGilbert1982}: $\widehat x_t$ depends only on $x_1,\dots,x_t$.
Its \emph{rate} from the initial state $A_0$ is the entropy rate
$\limsup_n n^{-1}H(S_{1:n})$ of the transmitted symbols, which block entropy
coding approaches~\cite{CoverThomas2006}. Since the symbol string determines a
reconstruction within $\eps$ of every coordinate, this rate is at least the
process $\eps$-entropy $R_\eps(X)$ of Section~\ref{sec.process}.

The simplest example has a single state. For an origin $\theta\in\R$, the
uniform quantizer with cells of width $2\eps$ has $\mathcal A=\{A_0\}$,
$S(A_0,x)=\lfloor (x-\theta)/(2\eps)+1/2\rfloor$, $d(A_0,k)=\theta+2\eps k$, and
$V(A_0,k)=A_0$. The adaptive scheme of Section~\ref{sec.dynamics} uses as its
state a finite set of stored reconstruction values: its cells are built around
the stored values, and each new reconstruction outside them replaces the oldest
stored value.

\section{Midpoint Representative Realization}\label{sec.dynamics}
\subsection{Construction}\label{sec.mrr}
MRR builds its cells around a set of stored values. A sample within $\eps$ of
a stored value is reconstructed by that value. Otherwise, the gap containing
the sample is divided into cells of width at most $2\eps$, anchored at the
neighboring stored values; the sample is reconstructed at the midpoint of its
cell, and that midpoint replaces the oldest stored value. 

For example, take $\eps=1$ and two stored values $0$ and $5$,
with $0$ the older value. An input $0.4$ is reconstructed as $0$,
and the state is unchanged. The uncovered interval $(1,4)$
is divided at $3$, with reconstruction values $2$ and $3.5$.
Thus an input $3.6$ is reconstructed as $3.5$: the encoder
sends the address of the second gap cell, and both encoder
and decoder replace $0$ by $3.5$. The decoder needs neither
the original input nor a separately transmitted reconstruction
value, since the shared state and the address determine the
midpoint.

We now formally define the construction.
\paragraph{States.}
MRR stores $\omega$ distinct reconstruction values, called representatives,
together with timestamps that record the order in which they were stored. Its
state is a set
\begin{equation}\label{eq.rep-state}
 A=\{(s_i,\tau_i):0\leq i<\omega\}\subset\R\times\N_0
\end{equation}
whose values $s_i$ are pairwise distinct and whose timestamps $\tau_i$ are
pairwise distinct. Write $\A_\omega$ for the collection of these states and
$\mathcal R_\omega$ for the collection of $\omega$-element subsets of $\R$.
The set of stored values,
\begin{equation}\label{eq.value-set}
 \val(A)=\{s:(s,\tau)\in A\text{ for some }\tau\},
\end{equation}
is the geometry of the state, written in increasing order as
$R=\{s_0<\cdots<s_{\omega-1}\}$. Sorting by value identifies $\A_\omega$ with a
Borel subset of $\R^\omega\times\N_0^\omega$. The cells and reconstructions
depend only on the geometry $R=\val(A)$; the timestamps serve only to identify
the oldest representative, which is the one replaced. In the notation of
Section~\ref{sec.ingredients}, $S(A,x)=S_R(x)$ and $d(A,s)=d_R(s)$, so that
$\pi_A=\pi_R$.

\paragraph{Address alphabet.}
Fix $R=\{s_0<\cdots<s_{\omega-1}\}$. The Midpoint
Representative Realization (MRR) uses the address alphabet
\[
 \begin{aligned}
 \mathcal S={}&\{(\mathrm P,i):0\leq i<\omega\}
 {}\cup\{(\mathrm I,i,j):0\leq i<\omega-1,\ j\in\N_0\}{}\cup\{(\mathrm L,j),(\mathrm R,j):j\in\N_0\}.
 \end{aligned}
\]
The priority and boundary conventions below are part of the definition.

\paragraph{Proximity.}
If $\dist(x,R)\leq\eps$, select the nearest representative, with ties
resolved in favor of the smaller value. If it is $s_i$, emit
$(\mathrm P,i)$ and reconstruct $s_i$. Denote its cell by $P_i(R)$.
It is the intersection of the closed halo $[s_i-\eps,s_i+\eps]$ with
the nearest-neighbor region assigned to $s_i$.

\paragraph{Interior refinement.}
For consecutive representatives put
\begin{equation}\label{eq.gap-parameters}
 a_i=s_i+\eps,\qquad b_i=s_{i+1}-\eps,
 \qquad d_i=b_i-a_i.
\end{equation}
If $d_i>0$, the uncovered gap is $G_i=(a_i,b_i)$. Set
\begin{equation}\label{eq.gap-grid}
 n_i=\ceil{\frac{d_i}{2\eps}},\qquad
 l_{ij}=a_i+2j\eps,\qquad
 u_{ij}=\min\{a_i+2(j+1)\eps,b_i\},
 \quad 0\leq j<n_i.
\end{equation}
For $x\in G_i$ choose $j=\floor{(x-a_i)/(2\eps)}$, emit
$(\mathrm I,i,j)$, and reconstruct
\begin{equation}\label{eq.interior-midpoint}
 q_{ij}=\frac{l_{ij}+u_{ij}}2.
\end{equation}
An internal cut belongs to the cell on its right. Both endpoints of $G_i$
belong to proximity cells. All gap cells except possibly the last have
length $2\eps$.

\paragraph{Exterior refinement.}
For $x>s_{\omega-1}+\eps$, put
\[
 j=\floor{\frac{x-s_{\omega-1}-\eps}{2\eps}},\qquad
 S_R(x)=(\mathrm R,j),\qquad
 \pi_R(x)=s_{\omega-1}+2(j+1)\eps.
\]
For $x<s_0-\eps$, put
\[
 j=\floor{\frac{s_0-\eps-x}{2\eps}},\qquad
 S_R(x)=(\mathrm L,j),\qquad
 \pi_R(x)=s_0-2(j+1)\eps.
\]
All exterior cells have length $2\eps$. At a halo endpoint, the proximity
rule takes precedence. These rules cover the case $\omega=1$ as well.

\paragraph{Replacement.}
For $A\in\A_\omega$ define
\begin{equation}\label{eq.clock}
 T(A)=\max\{\tau:(s,\tau)\in A\},\qquad
 (r_*(A),\tau_*(A))=\underset{(s,\tau)\in A}{\operatorname{argmin}}\ \tau.
\end{equation}
The state is unchanged on a proximity address. On every refinement address
with reconstruction $q$, set
\begin{equation}\label{eq.mrr-update}
 V(A,s)=\bigl(A\setminus\{(r_*(A),\tau_*(A))\}\bigr)
                \cup\{(q,T(A)+1)\}.
\end{equation}
For addresses unused at $\val(A)$ set $V(A,s)=A$. This convention defines
$V$ everywhere. Finally set $U(A,x)=V(A,S_{\val(A)}(x))$.
The timestamps form an insertion clock: they increase only when a
refinement occurs.

Figure~\ref{fig.framework} displays the reconstruction and update paths.
The encoder and decoder share $A_t$; the transmitted address determines
  the reconstruction and the next state through the maps $d$ and $V$.

\begin{figure}[tbp]
\centering
\includegraphics[width=\textwidth]{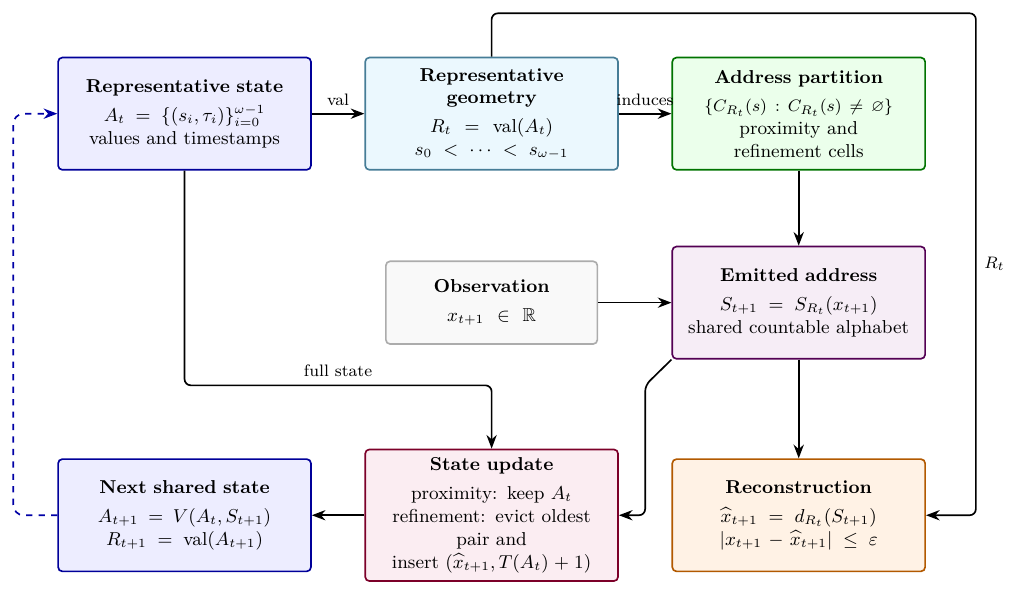}
\caption{Representative-state reconstruction and synchronized evolution
for MRR. The value projection determines the address partition. The full
timestamped state and the emitted address determine the update, with
$T(A_t)$ defined in \eqref{eq.clock}. Proximity cells reconstruct to their
stored representatives; refinement cells reconstruct to their midpoints.}
\label{fig.framework}
\end{figure}

\subsection{Structural properties}\label{sec.axioms}
MRR has four structural properties, each reflecting a design constraint.
\begin{enumerate}[label=(P\arabic*)]
\item \emph{Constant capacity.} Every state stores exactly $\omega$ values.  The number of stored reconstruction values is fixed in advance.
This is a capacity constraint, not a bound on storage in bits.
\item \emph{Distinct values.} The stored values are pairwise distinct, and so
 are their timestamps. Each stored value is a separate reconstruction point,
 and the oldest value is unique.
\item \emph{Time-invariant update.} The update is a single map $U(A,x)$: it
 does not depend on the time index, and it depends on past inputs only through
 the current state. For i.i.d.\ inputs this makes the state a time-homogeneous
Markov chain (Section~\ref{sec.compact}).
\item \emph{One replacement per step.} Each step either leaves the state
 unchanged or replaces one stored value by a new one. The shared state
 therefore changes by at most one value per transmitted symbol, and
 first-in, first-out replacement is well defined.
\end{enumerate}
Theorem~\ref{thm.realization} below verifies these properties, together with
decoder compatibility and~\eqref{eq.pathwise-error}.

\subsection{Well-posedness and separation}
\begin{theorem}[Pathwise realization]\label{thm.realization}
For every $\eps>0$, $\omega\geq1$, and $A_0\in\A_\omega$, MRR is a causal
quantizer with a shared state in the sense of Section~\ref{sec.ingredients}
  and satisfies properties (P1)--(P4). Every address
cell is an interval of positive length at most $2\eps$, and
\eqref{eq.pathwise-error} holds for every input sequence. Every refinement
inserts a value at distance strictly greater than $\eps$ from every value
in the preceding state. Moreover,
\begin{equation}\label{eq.idempotence}
 \pi_R(\pi_R(x))=\pi_R(x)\qquad(R\in\mathcal R_\omega, x\in\R).
\end{equation}
\end{theorem}
\begin{proof}
The proximity cells are intersections of intervals. Each contains an open
neighborhood of its representative because the representative values are
distinct. Where two halos overlap, the nearest-neighbor rule divides the
overlap at their midpoint, with a single specified owner of the tie.
Consequently, the proximity cells are disjoint intervals of positive
length, each contained in its representative's halo.

Their complement is the disjoint union of the nonempty open gaps $G_i$
and the two open exterior rays. If $x\in G_i$, then
$0<(x-a_i)/(2\eps)<d_i/(2\eps)$, so its prescribed index satisfies
$0\leq j<n_i$. The endpoints in \eqref{eq.gap-grid} satisfy
$0<u_{ij}-l_{ij}\leq2\eps$ for every such $j$. The floor rule and the
proximity priority assign every endpoint exactly once. The same argument
applies to the exterior rays. Thus all nonempty address cells are
intervals of positive length and form a countable partition of $\R$.

On proximity cells the error is at most $\eps$ by definition. On a
refinement cell the reconstruction is the midpoint of its two finite
endpoints, and the length is at most $2\eps$. The error bound follows,
including at any assigned endpoint.

For an interior refinement,
$s_i+\eps<q_{ij}<s_{i+1}-\eps$. Its distance from either neighboring
representative, and hence from every representative, is strictly greater
than $\eps$. An exterior midpoint is at distance at least $2\eps$ from
the nearest extreme representative. Thus the inserted value is new.
The inserted timestamp $T(A)+1$ exceeds every current timestamp. Removing
one pair and inserting this new pair therefore preserves both
distinctness conditions and the capacity. All rules are autonomous, and
the only possible change is the replacement in \eqref{eq.mrr-update}.
Hence MRR has properties (P1)--(P4), and its update depends only on the
current state and the transmitted symbol, as required
in~\eqref{eq.shared-state}.

Sorting distinct coordinates, comparison with a finite collection of
values, the floor and ceiling functions, and the finite minimum and
maximum operations are Borel. Each address class is described by Borel
inequalities. These observations establish the measurability,
including that of the update on unused addresses. Finally, a proximity
reconstruction is a representative and is fixed by $\pi_R$. A refinement
midpoint lies in the interior of its own refinement cell, so it has the
same address and reconstruction when quantized again. This proves
\eqref{eq.idempotence}.
\end{proof}

Define the number of refinements through time $t$ by
\begin{equation}\label{eq.refinement-count}
 N_t=\sum_{j=1}^t\1_{\{A_j\ne A_{j-1}\}}.
\end{equation}
The insertion clock obeys the exact identity
\begin{equation}\label{eq.clock-count}
 T(A_t)=T(A_0)+N_t.
\end{equation}
In particular, no restriction relating initial timestamps to the external
time index is needed.

\begin{proposition}[Separation and the insertion trace]\label{prop.separation}
Let $\delta_0$ be the minimum distance between distinct initial
representatives, with $\delta_0=+\infty$ if $\omega=1$, and put
$\delta=\min\{\delta_0,\eps\}$. Then every geometry along every input
path is $\delta$-separated. If $R_0$ is $\eps$-separated, this property
holds at all times. After $\omega$ refinements have occurred, every pair
of current representatives has distance strictly greater than $\eps$.

Let $q_k$ be the numeric reconstruction inserted at the $k$th refinement.
After $k\geq\omega$ refinements, the state consists exactly of the pairs
\begin{equation}\label{eq.fifo-trace}
 \{(q_j,T(A_0)+j):k-\omega+1\leq j\leq k\}.
\end{equation}
Before then, it consists of the inserted pairs and the $\omega-k$ initial
pairs with largest timestamps.
\end{proposition}
\begin{proof}
On a proximity event the state does not change. At a refinement, all
distances between retained values are unchanged, and the new value has
distance greater than $\eps$ from every retained value by
Theorem~\ref{thm.realization}. Induction proves the first property.

The first inserted timestamp exceeds all initial timestamps. At each
subsequent refinement the inserted timestamp exceeds all surviving ones.
Thus the eviction order is the initial timestamp order, followed by the
insertion order. This proves \eqref{eq.fifo-trace} and the assertion for
$k<\omega$. Once no initial pair remains, consider any two current
values in their insertion order. The later value was inserted while the
earlier one was present, so their distance is greater than $\eps$.
\end{proof}

The number of observations between refinements is unrestricted. In
particular, the time of the $\omega$th refinement can be infinite.
Equation~\eqref{eq.fifo-trace} concerns numeric insertions. Local addresses
are interpreted relative to the preceding geometry, so the last
$\omega$ addresses alone need not determine the current state.

\subsection{Synchronization and local stability}
\begin{proposition}[Synchronization and stable address paths]\label{prop.synchronization}
Starting from the shared state $A_0$, the decoder recovers all
reconstructions and states from the emitted addresses. At every finite
time, the set of reachable states is countable.

For a fixed input string $x_1,\ldots,x_n$, let $C_t$ be its address cell
at time $t$, and set
\begin{equation}\label{eq.stability-margin}
 m_t=\dist(x_t,\R\setminus C_t).
\end{equation}
Any string $y_1,\ldots,y_n$ satisfying $|y_t-x_t|<m_t$ for all $t$
produces the same addresses, reconstructions, and states through time $n$.
\end{proposition}
\begin{proof}
Given $A_{t-1}$ and $s_t$, the decoder computes
$d_{\val(A_{t-1})}(s_t)$ and $V(A_{t-1},s_t)$. Induction gives
synchronization. A finite string over a countable alphabet has only
countably many possible values. Since $A_t$ is a function of that string
and $A_0$, its range is countable.

For the second assertion, the two systems have the same initial state.
Suppose their states agree before step $t$. The strict inequality in the
hypothesis implies $y_t\in C_t$, so the addresses agree at that step.
Decoder compatibility then gives the same reconstruction and next state.
Induction completes the proof. If $m_t=0$, the displayed strict
perturbation condition has no solution at that coordinate; positive
margins are the substantive case.
\end{proof}

A replacement also has a local effect on the ordered representative
gaps. Deleting a representative merges its two adjacent gaps, where an
extreme gap may be unbounded. Inserting the new value splits one gap of
the remaining geometry. Hence at most three gaps of the original
geometry are affected. This locality does not imply monotonic decrease
of the maximum gap: deletion can merge two large neighboring intervals.

\subsection{Bounded inputs and compact age-ranked dynamics}\label{sec.compact}
\begin{proposition}[Trapping]\label{prop.trapping}
Suppose $x_t\in[L_0,U_0]$ for every $t$. Put
\begin{equation}\label{eq.trapping-intervals}
 K_\eps=[L_0-\eps,U_0+\eps],\qquad
 J=\operatorname{conv}(R_0\cup K_\eps).
\end{equation}
All representatives remain in $J$. Every inserted value belongs to
$K_\eps$, and after $\omega$ refinements all representatives belong to
$K_\eps$.
\end{proposition}
\begin{proof}
At a refinement the inserted value equals the reconstruction of the
current input. Theorem~\ref{thm.realization} puts it within $\eps$ of
$[L_0,U_0]$, hence in $K_\eps$. Retained representatives are unchanged,
which proves invariance of $J$ by induction. The last assertion follows
from Proposition~\ref{prop.separation} once every initial pair has been
evicted.
\end{proof}

Order the values by increasing timestamp and write
$a=(r_1,\ldots,r_\omega)$, with $r_1$ oldest. Numerical timestamp values
are irrelevant to this ordering. The induced update is
\begin{equation}\label{eq.age-update}
 \Phi(a,x)=
 \begin{cases}
 a,&\text{on a proximity event},\\
 (r_2,\ldots,r_\omega,\pi_{\{r_1,\ldots,r_\omega\}}(x)),
       &\text{on a refinement event}.
 \end{cases}
\end{equation}
For $\omega=1$, the second line is the one-coordinate tuple containing
the new value. With $J$ and $\delta$ as above, define
\begin{equation}\label{eq.compact-state}
 \mathcal K_{J,\delta}
 =\{a\in J^\omega:|r_i-r_j|\geq\delta\text{ for }i\ne j\}.
\end{equation}
This is a compact state space invariant under all inputs in $[L_0,U_0]$.

\begin{theorem}[Feller age-ranked dynamics]\label{thm.feller}
Let the inputs be independent with a common atomless law $\nu$ supported
on $[L_0,U_0]$. On $\mathcal K_{J,\delta}$, the age-ranked process is a
time-homogeneous Markov chain with transition kernel
\begin{equation}\label{eq.feller-kernel}
 P(a,B)=\int\1_B(\Phi(a,x))\,\nu(dx).
\end{equation}
This kernel is Feller and has at least one invariant probability measure.
\end{theorem}
\begin{proof}
The update depends only on the current age-ranked values and the new
input. Independence of the latter gives the Markov property and
\eqref{eq.feller-kernel}. The measurability follows from
Theorem~\ref{thm.realization}. Invariance of the compact state space
follows from the separation and trapping propositions.

Fix $a\in\mathcal K_{J,\delta}$ and a sequence $a^{(m)}\to a$ in this
space. The spatial order of the coordinates is eventually the same as
at $a$, since distinct coordinates of $a$ are separated by $\delta>0$.
Write their spatially ordered values as $s_i^{(m)}\to s_i$.
We show that
\begin{equation}\label{eq.ae-continuity}
 \Phi(a^{(m)},x)\longrightarrow\Phi(a,x)
 \quad\text{for all }x\in[L_0,U_0]\setminus E_a,
\end{equation}
where $E_a$ is finite.

Include in $E_a$ every halo endpoint $s_i\pm\eps$, every midpoint
$(s_i+s_{i+1})/2$, all interior cuts $s_i+\eps+2j\eps$ lying in
$[s_i+\eps,s_{i+1}-\eps]$, and all exterior cuts in
$[L_0,U_0]$. There are finitely many such points, since both the state
interval and the input interval are bounded and $\eps>0$.

For $x\notin E_a$ the relevant inequalities selecting proximity,
an uncovered gap, or an exterior ray are strict and remain valid for
large $m$. On a proximity event the selected representative is stable,
and the update in \eqref{eq.age-update} is the identity. On a refinement
event the floor index is eventually constant. Its reconstruction is a
continuous function of the relevant representative coordinates: for an
interior cell this follows from the minimum in \eqref{eq.gap-grid}, and
for an exterior cell it follows from its explicit formula. If a gap
length crosses a multiple of $2\eps$, the additional cell has vanishing
length at the right endpoint; the anchored cuts before that endpoint
do not move to new indices. Thus a fixed $x\notin E_a$ is unaffected
by that change for large $m$. The age order fixes the deleted coordinate,
which proves \eqref{eq.ae-continuity}.

For a continuous function $g$ on $\mathcal K_{J,\delta}$, compactness
implies boundedness. Atomlessness gives $\nu(E_a)=0$. Dominated
convergence applied to \eqref{eq.ae-continuity} yields
$Pg(a^{(m)})\to Pg(a)$. Hence $P$ is Feller. The existence of an
invariant probability now follows from the compact-state
Krylov--Bogoliubov theorem; see \cite[Chapter~12]{MeynTweedie2009}.
\end{proof}

Invariant measures need not be unique. For example, with $\eps=1$,
$\omega=2$, and inputs uniform on $[0,1]$, each of the geometries
$\{0,1\}$ and $\{-1/4,3/4\}$ covers the input interval by proximity
halos. Either fixed age ordering is absorbing. Numerical timestamps are
unnecessary for the compact formulation: when infinitely many refinements
occur, their maximum increases without bound by \eqref{eq.clock-count}.
\begin{remark}[Finite reconstruction alphabets]
\label{rem.finite-reconstruction}
Suppose every refinement reconstruction belongs to a finite set $G$.
Put $H=G\cup R_0$ and $m=|H|$. The age-ranked state then takes values
among the ordered $\omega$-tuples of distinct elements of $H$, of
which there are
\[
N=\frac{m!}{(m-\omega)!}.
\]
For i.i.d.\ inputs, this gives a finite-state time-homogeneous Markov
chain. For a deterministic input of period $p$, augmenting the state
with the input phase gives a deterministic map on at most $pN$
states. Hence the age-ranked states, reconstructions, and addresses
are eventually periodic, with preperiod and period each at most $pN$.
For example, for inputs in $[L_0,U_0]$ with $R_0\subset\Lambda_{\theta,\varepsilon}$, the finite-alphabet hypothesis holds:
by Propositions~\ref{prop.trapping} and~\ref{prop.lattice}, every inserted value
lies in the finite set $\Lambda_{\theta,\varepsilon}\cap[L_0-\varepsilon,U_0+\varepsilon]$.
\end{remark}
\section{Geometry of the adaptive partition}\label{sec.overlap}
\subsection{Minimality with prescribed proximity regions}
For a representative gap of width $W=s_{i+1}-s_i>0$, the uncovered
length is $(W-2\eps)_+=\max\{W-2\eps,0\}$. The proximity regions are
prescribed by Section~\ref{sec.mrr}; minimality below is conditional on
retaining these regions and their reconstruction values.

\begin{proposition}[Constrained local cell count]\label{prop.local-count}
On $[s_i,s_{i+1}]$, the minimum number of additional cells of diameter at
most $2\eps$ needed to cover the part outside the prescribed proximity
regions is
\begin{equation}\label{eq.additional-count}
 n_i=\ceil{\frac{(W-2\eps)_+}{2\eps}}.
\end{equation}
MRR attains this minimum. Counting the two nonempty proximity pieces,
the number of address cells intersecting this representative interval is
\begin{equation}\label{eq.local-count}
 N_i=2+n_i=1+\ceil{\frac{W}{2\eps}}.
\end{equation}
\end{proposition}
\begin{proof}
If $W\leq2\eps$, the two nearest-representative proximity pieces cover
the interval, so no additional cell is required. If $W>2\eps$, their
complement in the interval is the open gap of length $W-2\eps$.
The elementary one-dimensional covering bound underlying
epsilon-capacity \cite{KolmogorovTikhomirov1959} requires at least the
number of cells in \eqref{eq.additional-count}. The cells in \eqref{eq.gap-grid}
give exactly that many intervals and cover the gap with the stipulated
endpoint convention. Adding the two proximity pieces and using
$\ceil{z-1}=\ceil z-1$ for $z>1$ gives \eqref{eq.local-count}.
\end{proof}

Without the prescribed proximity pieces, the covering number of a
positive-length interval of width $W$ by intervals of diameter at most
$2\eps$ is $\ceil{W/(2\eps)}$ \cite{KolmogorovTikhomirov1959}.
The constraint in Proposition~\ref{prop.local-count} is therefore
essential. The remainder cell can have any length in $(0,2\eps]$.

\subsection{Four-cell overlap}
Let $\mathcal P_R$ be the MRR address partition. For an interval $J$,
count cells by nonempty intersection, including any assigned endpoints.
Define
\begin{equation}\label{eq.overlap-number}
 \kappa_\eps(R)=
 \sup_{J:\,\diam J\leq2\eps}
 \#\{C\in\mathcal P_R:C\cap J\ne\varnothing\}.
\end{equation}

\begin{lemma}[Cell widths and short cells]\label{lem.width-structure}
For every geometry, at most
\begin{equation}
m_\omega:=\left\lfloor\frac{3\omega}{2}\right\rfloor-1 \label{eq.momega}    
\end{equation}
address cells have length strictly less than $2\varepsilon$. This bound is sharp, even among $\varepsilon$-separated geometries. If the geometry is $\eps$-separated, every
proximity cell has length at least $\eps$, and every three consecutive
address cells have total length strictly greater than $2\eps$.
\end{lemma}
\begin{proof}
Let $c$ be the number of consecutive representative gaps smaller
than $2\varepsilon$. A proximity cell is short only if its
representative is incident to one of these gaps. Hence there are
at most $\min\{2c,\omega\}$ short proximity cells.

The $c$ gaps contain no refinement cells. Each of the remaining
$\omega-1-c$ gaps contributes at most one short refinement cell,
and all exterior cells have full length. Consequently,
\[
N_{\mathrm{short}}
\le \min\{2c,\omega\}+\omega-1-c
\le \left\lfloor\frac{3\omega}{2}\right\rfloor-1.
\]
When $\omega\ge2$, we choose consecutive representative
gaps alternately equal to $3\varepsilon/2$ and $3\varepsilon$,
starting with $3\varepsilon/2$. The shorter gaps give
$2\lfloor\omega/2\rfloor$ short proximity cells, and the longer
gaps give $\lfloor(\omega-1)/2\rfloor$ short refinement cells.
Their sum is $m_\omega$. This geometry is
$\varepsilon$-separated. For $\omega=1$, all cells have full
length and $m_1=0$.

For an interior representative, its proximity interval extends to each
side by
$\min\{\eps,(s_i-s_{i-1})/2\}$ and
$\min\{\eps,(s_{i+1}-s_i)/2\}$, respectively. Under separation each
extent is at least $\eps/2$. An extreme proximity interval has extent
$\eps$ on its unbounded side and at least $\eps/2$ on its other side.
For $\omega=1$ its length is $2\eps$. Endpoint assignments do not
change these lengths.

Consider three consecutive cells. If one is a full-width refinement
cell, the other two have positive length, and the sum exceeds $2\eps$.
If at least two are proximity cells, their combined length is at least
$2\eps$, and the third has positive length. The only remaining
possibility is two short refinement cells and a single proximity cell.
Two short refinement cells cannot be adjacent: each gap contains at most
one, and distinct gaps are separated by a proximity cell. Hence the
pattern must be remainder--proximity--remainder. The middle
representative then has an uncovered gap on both sides, so neither side
of its halo is clipped by a neighboring halo. Its proximity length is
$2\eps$, and the two remainders have positive length. This proves the
last assertion.
\end{proof}

\begin{theorem}[Sharp four-cell overlap]\label{thm.four-overlap}
If $R$ is $\eps$-separated, then $\kappa_\eps(R)\leq4$.
The constant four cannot be decreased uniformly over separated
geometries and capacities.
\end{theorem}
\begin{proof}
If an interval $J$ meets five cells, it meets five consecutive cells
between its first and last intersections. The three middle cells lie
between a point of the first cell and a point of the fifth. Consequently
their total length is at most $\diam J\leq2\eps$. This contradicts
Lemma~\ref{lem.width-structure}. The argument includes endpoint-only
intersections because all cells are intervals of positive length.

For sharpness take $\eps=1$ and $R=\{0,2.1,3.1\}$. The four points
$0.99$, $1.05$, $2$, and $2.61$ lie respectively in the proximity cell
of $0$, the remainder gap $(1,1.1)$, the proximity cell of $2.1$, and
the proximity cell of $3.1$. Their interval hull has length $1.62<2$.
Scaling the configuration gives the same conclusion for every
$\eps>0$; see Figure~\ref{fig.overlap}.
\end{proof}

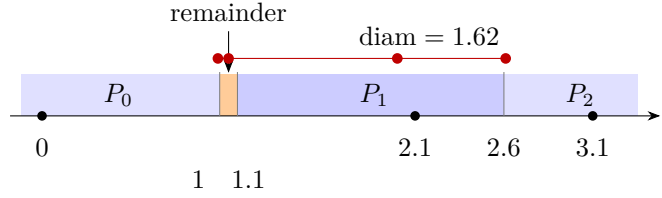
\begin{figure}[tb]
\centering
\begin{tikzpicture}[x=2.35cm,y=1cm,>=Stealth,font=\small]
\fill[blue!12] (-0.12,0) rectangle (1,0.55);
\fill[orange!40] (1,0) rectangle (1.1,0.55);
\fill[blue!20] (1.1,0) rectangle (2.6,0.55);
\fill[blue!12] (2.6,0) rectangle (3.35,0.55);
\draw[->] (-0.18,0)--(3.48,0);
\foreach \x in {1,1.1,2.6} {\draw[gray] (\x,0)--(\x,0.58);}
\foreach \x/\lab in {0/0,2.1/2.1,3.1/3.1} {
 \fill (\x,0) circle (1.8pt);
 \node[below=5pt] at (\x,0) {$\lab$};
}
\node at (0.43,0.28) {$P_0$};
\node at (1.87,0.28) {$P_1$};
\node at (3.03,0.28) {$P_2$};
\draw[->] (1.05,1.12)--(1.05,0.57);
\node[above] at (1.05,1.12) {remainder};
\foreach \x in {0.99,1.05,2,2.61} {\fill[red!75!black] (\x,0.76) circle (2pt);}
\draw[red!75!black] (0.99,0.76)--(2.61,0.76);
\node[above] at (2.18,0.8) {$\operatorname{diam}=1.62$};
\node[below=5pt] at (1.05,-0.42) {$1\quad1.1$};
\node[below=5pt] at (2.6,0) {$2.6$};
\end{tikzpicture}
\caption{A separated geometry with a short remainder cell. With
$\eps=1$ and $R=\{0,2.1,3.1\}$, the four marked inputs have different
addresses and lie in an interval of length $1.62$. The remainder has
length $0.1$; the neighboring proximity cells preserve the four-cell
overlap bound. Endpoint ownership is specified in Section~\ref{sec.mrr}.}
\label{fig.overlap}
\end{figure}

\subsection{Scalar $\varepsilon$-entropy at equal conditioning}\label{sec:overlap}
For a probability measure $\mu$ on $\R$, define the scalar $\varepsilon$-entropy:
\begin{equation}\label{eq.scalar-floor}
 F_\eps(\mu)=\inf_Q H_\mu(Q(X)),
\end{equation}
where $Q$ ranges over countable Borel partitions equipped with a
reconstruction value for each cell such that the error is at most
$\eps$ everywhere. Equivalently, each cell must be contained in a closed
interval of radius $\eps$. The infimum is allowed to be infinite.

If every cell of a countable partition $\mathcal Q$ intersects at most
$k$ cells of a countable partition $\mathcal P$, then, conditional on
the $\mathcal Q$-cell, at most $k$ labels of $\mathcal P$ are possible.
The entropy chain rule therefore gives
\begin{equation}\label{eq.partition-overlap}
 H(\mathcal P)\leq H(\mathcal P,\mathcal Q)
 =H(\mathcal Q)+H(\mathcal P\mid\mathcal Q)
 \leq H(\mathcal Q)+\log_2 k.
\end{equation}
The same comparison appears, with its companion estimate for
translated partitions, among the basic properties of entropy at a
fixed scale used in fractal geometry~\cite[Lemma~3.2]{Hochman2014}.

\begin{theorem}[Sharp two-bit scalar comparison]\label{thm.two-bit}
For every $\eps$-separated geometry $R$ and every probability law $\mu$,
\begin{equation}\label{eq.two-bit}
 F_\eps(\mu)\leq H_\mu(S_R(X))\leq F_\eps(\mu)+2.
\end{equation}
The upper bound is sharp already for a three-representative geometry
and a finitely supported law with $F_\eps(\mu)=0$.

More generally, let $B$ be countably valued, let $R(B)$ be
$\eps$-separated for every value of $B$ having positive probability,
and let $\mu_b$ be the conditional law of $X$ given $B=b$. Then
\begin{equation}\label{eq.conditional-two-bit}
 \sum_b\Prob(B=b)F_\eps(\mu_b)
 \leq H(S_{R(B)}(X)\mid B)
 \leq\sum_b\Prob(B=b)F_\eps(\mu_b)+2.
\end{equation}
\end{theorem}
\begin{proof}
The lower bound holds because the MRR partition is admissible in
\eqref{eq.scalar-floor}. Every cell of an admissible competitor is
contained in an interval of diameter $2\eps$. By
Theorem~\ref{thm.four-overlap} it intersects at most four MRR cells.
Inequality~\eqref{eq.partition-overlap} gives
$H_\mu(S_R(X))\leq H_\mu(Q(X))+2$ for every competitor. Taking the
infimum proves \eqref{eq.two-bit}, with the extended-real convention
when $F_\eps(\mu)=\infty$.

For sharpness use the geometry in Theorem~\ref{thm.four-overlap} and
assign mass $1/4$ to each of its four displayed points. Their four
addresses have entropy two bits. Their common interval hull has radius
$0.81<1$, so one faithful cell contains all the mass; the rest of the
line can be partitioned into zero-probability faithful cells. Therefore
the scalar $\varepsilon$-entropy is zero. Finally, apply \eqref{eq.two-bit} separately
under each conditional law $\mu_b$ and sum with the probabilities of
$B=b$. Countability makes these sums unambiguous.
\end{proof}

With deterministic initialization, Proposition~\ref{prop.synchronization}
allows $B=A_{t-1}$ or $B=(S_1,\ldots,S_{t-1})$ at every finite time.
The comparison holds throughout a trajectory initialized with separated
representatives. For arbitrary initialization it applies on the states
reached after at least $\omega$ refinements. These are statements at a
specified conditioning. No identification of the representative state
with the full source past is required or implied.

\section{Entropy rates}\label{sec.process}
At a fixed state, geometric complexity is the number and arrangement of
the reconstruction cells. By Proposition~\ref{prop.local-count}, the representative hull
$[s_0,s_{\omega-1}]$ meets exactly
$\omega+\sum_{i=0}^{\omega-2}\lceil(\Delta_i-2\varepsilon)_+/(2\varepsilon)\rceil$
cells, where $\Delta_i=s_{i+1}-s_i$. For a source law, the probabilities of these cells
determine symbolic entropy. This section relates their changing geometry
to block entropy and to the process $\eps$-entropy. The decisive geometric facts
are that address preimages are rectangles and that at most
$m_\omega$ scalar cells per state have less than full width.

\subsection{The block benchmark}\label{sec.block-floor}
For an $\R^n$-valued random vector $X_{1:n}$, define the \emph{block floor}
\begin{equation}\label{eq.block-floor}
 C_{n,\eps}(X)=\inf_Q H(Q(X_{1:n})),
\end{equation}
where $Q$ is a countable Borel partition of $\R^n$ and each of its cells
has a reconstruction vector $y$ satisfying
$\|x-y\|_\infty\leq\eps$ throughout that cell. We use closed error
balls. Null-probability parts of the domain may be covered by additional
cells without changing entropy.

For a stationary process with $C_{1,\eps}<\infty$, concatenating
admissible block partitions makes $C_{n,\eps}$ subadditive. Consequently
the process $\varepsilon$-entropy is
\begin{equation}\label{eq.process-floor}
 R_\eps(X)=\lim_{n\to\infty}\frac{C_{n,\eps}(X)}n
          =\inf_{n\geq1}\frac{C_{n,\eps}(X)}n.
\end{equation}
This is the entropy-of-partitions formulation of the bounded-error
process $\eps$-entropy \cite{Posner1967,posner1971,Bucklew1991}.    Countable
variable-length block coding approaches this entropy per coordinate by
the noiseless coding theorem \cite{CoverThomas2006}. The convention
here enforces the error bound on every reconstructed block; it does not
require a finite codebook covering an unbounded input space.

If $X_{1:n}$ has a density $f_n$ and finite differential entropy, the
standard maximum-entropy bound for a set contained in a cube of volume
$(2\eps)^n$ gives
\begin{equation}\label{eq.shannon-lower}
 C_{n,\eps}(X)\geq h(X_{1:n})-n\log_2(2\eps).
\end{equation}
See \cite[Corollary~1]{Bucklew1991} and the conditional-entropy
formulation in \cite{CoverThomas2006}. We call the right-hand side the
Shannon lower bound.
A conditional form of this bound, with rate measured by directed mutual
information, is the main tool in the rate--cost analysis of feedback control in~\cite{KostinaHassibi2019}.

For MRR initialized at a fixed shared $A_0$, write $S_t$ for the address
at time $t$ and define
\begin{equation}\label{eq.mrr-rate}
 R_{\mathrm{MRR}}(X;A_0)=
 \limsup_{n\to\infty}\frac1n H(S_{1:n}).
\end{equation}
Every address string specifies a faithful reconstruction vector, so
$C_{n,\eps}(X)\leq H(S_{1:n})$ and
$R_\eps(X)\leq R_{\mathrm{MRR}}(X;A_0)$.
The chain rule uses $H(S_t\mid S_{1:t-1})$. Replacing the full address
history by $A_{t-1}$ generally increases this conditional entropy.

\paragraph{Causal benchmark.}
A \emph{causal code} with tolerance $\varepsilon$ consists of Borel maps
$g_t:\mathbb{R}^t\to\mathcal{L}$ into a countable label set and maps
$\phi_t:\mathcal{L}^t\to\mathbb{R}$ such that
\[
  \bigl|x_t-\phi_t\bigl(g_1(x_1),\dots,g_t(x_{1:t})\bigr)\bigr|\le\varepsilon
  \qquad (t\ge 1)
\]
for every input sequence. Its labels are $L_t=g_t(X_{1:t})$, its rate is
$\limsup_{n} n^{-1}H(L_{1:n})$, and the labels may be entropy-coded in blocks,
as in~\cite{NeuhoffGilbert1982}. We call the infimum $R^{\mathrm{c}}_\varepsilon(X)$ of these rates  the {\it causal floor}. Each label block determines a faithful reconstruction
vector, so $C_{n,\varepsilon}(X)\le H(L_{1:n})$ and
$R_\varepsilon(X)\le R^{\mathrm{c}}_\varepsilon(X)$. By decoder compatibility,
MRR is a causal code with $L_t=S_t$. Hence
\begin{equation}\label{eq.causal-sandwich}
  R_\varepsilon(X)\le R^{\mathrm{c}}_\varepsilon(X)\le R_{\mathrm{MRR}}(X;A_0).
\end{equation}
 
\subsection{A distribution-free process comparison}\label{sec.causal-overlap}
The scalar overlap bound   controls entire address strings. The
relevant comparison is between block partitions, so no identification of
the representative state with the source history is needed.

\begin{proposition}[Overlap along a causal address tree]\label{prop.causal-overlap}
Consider a shared deterministic initial state of an $\eps$-faithful causal quantizer with a shared state (Section~\ref{sec.ingredients}). Suppose that an integer $m\geq1$ has the
following property: at every state reachable from this initial state,
each closed interval of length $2\eps$ intersects at most $m$ nonempty
address cells. Then every faithful block cell in $\R^n$ intersects at
most $m^n$ cells of the system's length-$n$ address partition.
Consequently, for every $\R^n$-valued random vector,
\begin{equation}\label{eq.causal-overlap-block}
 C_{n,\eps}(X)\leq H(S_{1:n})
                 \leq C_{n,\eps}(X)+n\log_2m.
\end{equation}
If $X$ is stationary and $C_{1,\eps}(X)<\infty$, then
\begin{equation}\label{eq.causal-overlap-rate}
 0\leq\limsup_{n\to\infty}\frac{H(S_{1:n})}{n}-R_\eps(X)
 \leq\log_2m.
\end{equation}
\end{proposition}
\begin{proof}
Let $B$ be a nonempty cell of a competing faithful block quantizer,
with reconstruction $y=(y_1,\ldots,y_n)$. Then
\[
 B\subseteq\prod_{t=1}^n I_t,\qquad
 I_t=[y_t-\eps,y_t+\eps].
\]
For $0\leq t\leq n$, let $\mathcal T_t(B)$ be the set of address
prefixes of length $t$ produced by some input block in $B$. There is
one empty prefix at $t=0$. Fix a prefix in $\mathcal T_{t-1}(B)$.
Decoder compatibility determines a unique state from that prefix and
the initial state. Any extension realized by a point of $B$ must select
an address cell at this state that intersects $I_t$. There are at most
$m$ such cells, and hence at most $m$ extensions. Induction gives
$|\mathcal T_t(B)|\leq m^t$. At $t=n$ this is the asserted overlap
bound, even when $B$ is not a rectangle.

Apply \eqref{eq.partition-overlap} to the address partition and
an arbitrary faithful block partition $Q$. The bound just proved gives
$H(S_{1:n})\leq H(Q(X_{1:n}))+\log_2(m^n)$. Taking the infimum
over $Q$ proves the upper bound in \eqref{eq.causal-overlap-block}.
Both bounds hold in the extended reals; in particular, $H(S_{1:n})$ is finite
whenever $C_{n,\eps}(X)$ is.
The lower bound follows
because the address string determines a faithful reconstruction vector.

If $X$ is stationary and $C_{1,\eps}(X)<\infty$, concatenation gives
$C_{n,\eps}\leq nC_{1,\eps}<\infty$. Divide
\eqref{eq.causal-overlap-block} by $n$ and use
\eqref{eq.process-floor}. Taking the limsup of the middle term proves
\eqref{eq.causal-overlap-rate}.
\end{proof}

\begin{corollary}[Two-bit process comparison]\label{cor.two-bit-process}
Initialize MRR with an $\eps$-separated geometry. For every block law,
\begin{equation}\label{eq.two-bit-process-block}
 C_{n,\eps}(X)\leq H(S_{1:n})\leq C_{n,\eps}(X)+2n.
\end{equation}
For every stationary source with $C_{1,\eps}(X)<\infty$,
\begin{equation}\label{eq.two-bit-process-rate}
 0\leq R_{\mathrm{MRR}}(X;A_0)-R_\eps(X)\leq2.
\end{equation}
No density, independence, or mixing assumption is required.
\end{corollary}
\begin{proof}
Proposition~\ref{prop.separation} preserves $\eps$-separation at
every reachable state. Theorem~\ref{thm.four-overlap} bounds the
number of cells intersecting each closed interval of length $2\eps$
by four. Apply Proposition~\ref{prop.causal-overlap} with $m=4$.
\end{proof}

The separation assumption cannot be dropped from a bound uniform in
capacity. For $\eps=1$, take $\omega$ distinct initial values in
$[0,1]$ and let the source be independent and uniform on these values.
Every observation is reconstructed by its own stored representative,
so the state remains fixed and $H(S_{1:n})=n\log_2\omega$.
The constant reconstruction vector $(1/2,\ldots,1/2)$ is faithful on
the support, so $C_{n,1}(X)=0$. The rest of the input space can be
partitioned into additional zero-probability faithful cells.

\subsection{Adaptive reconstruction rectangles}
For a valid address string $s_{1:n}$, decoder compatibility determines
the preceding state at each step. Let $C_t(s_{1:t})$ be the scalar cell
selected at that state and let $\ell_t(s_{1:t})$ be its length. The
preimage of the string is exactly
\begin{equation}\label{eq.adaptive-rectangle}
 C(s_{1:n})=\prod_{t=1}^n C_t(s_{1:t}),\qquad
 v(s_{1:n})=\prod_{t=1}^n\ell_t(s_{1:t}).
\end{equation}
Indeed, after fixing a prefix, the next state is fixed; the restriction
on the next input is membership in its specified scalar cell.
Induction over the prefix proves the product description. Thus adaptation
changes the collection of rectangles without changing the rectangular
shape of an individual address preimage.

For a block density $f_n$, put $p_s=\Prob(S_{1:n}=s)$ and define its
cellwise average by
\begin{equation}\label{eq.cell-average}
 \overline f_n(x)=\frac{p_s}{v(s)}\quad(x\in C(s)).
\end{equation}
Choices on rectangle boundaries have Lebesgue measure zero.

\begin{proposition}[Exact adaptive entropy decomposition]\label{prop.rectangle-identity}
Suppose $X_{1:n}$ has density $f_n$ with
$\int f_n|\log_2 f_n|<\infty$, and suppose
$H(S_{1:n})<\infty$ and $\E|\log_2 v(S_{1:n})|<\infty$.
Then
\begin{equation}\label{eq.entropy-decomposition}
 H(S_{1:n})-h(X_{1:n})+n\log_2(2\eps)
 =D_2(f_n\|\overline f_n)
   +\sum_{t=1}^n\E\log_2\frac{2\eps}{\ell_t(S_{1:t})}.
\end{equation}
Both terms on the right are nonnegative. They vanish simultaneously
exactly when every positive-probability rectangle has full side lengths
$2\eps$ and $f_n$ is constant almost everywhere on each such rectangle.
\end{proposition}
\begin{proof}
By \eqref{eq.adaptive-rectangle}, the rectangles form a countable
partition up to their assigned boundaries, and each has positive finite
volume. On a rectangle of positive probability, the logarithm of
\eqref{eq.cell-average} is $\log_2 p_s-\log_2 v(s)$. Therefore
\[
 \int f_n\log_2\overline f_n
 =\sum_{s:p_s>0}p_s\log_2p_s
   -\sum_{s:p_s>0}p_s\log_2v(s)
 =-H(S_{1:n})-\E\log_2v(S_{1:n}).
\]
The stated integrability assumptions justify these sums and integrals.
Subtracting this expression from $\int f_n\log_2 f_n=-h(X_{1:n})$
and using the product formula for $v$ gives
\eqref{eq.entropy-decomposition}. Relative entropy is nonnegative,
with equality only for almost-everywhere equality of its two densities
\cite{CoverThomas2006}. Also $0<\ell_t\leq2\eps$, so each volume
penalty is nonnegative and vanishes precisely for a full-width side.
These equality conditions give the last assertion.
\end{proof}

\subsection{A quantitative density estimate}
For a density $f$ on $\R^n$, denote its Fisher information by
\begin{equation}\label{eq.block-fisher-definition}
 J(f)=4\int_{\R^n}|\nabla\sqrt f|^2\,dx,
 \qquad \sqrt f\in H^1(\R^n).
\end{equation}
Here $H^1$ denotes the first-order Sobolev space and  derivatives
are weak derivatives. For a positive differentiable density this is
$J(f)=\int f\|\nabla\ln f\|^2\,dx$; the score uses natural
logarithms even though entropies are measured in bits. Densities are allowed to vanish.

\begin{lemma}[Fisher information inside a rectangle]
\label{lem.rectangle-kl}
Let $C=\prod_{i=1}^n I_i$ have positive finite side lengths
$\ell_i=|I_i|$. Let $f\geq0$ satisfy
$p_C=\int_C f\,dx\in(0,\infty)$ and $\sqrt f\in H^1(C)$.
Write $P_C$ for the probability law with density $f/p_C$ on $C$,
and $U_C$ for the uniform law on $C$. Then
\begin{equation}\label{eq.rectangle-kl}
 D_2(P_C\|U_C)
 \leq\frac{2}{\pi^2\ln2\,p_C}
       \sum_{i=1}^n\ell_i^2\int_C|\partial_i\sqrt f|^2\,dx.
\end{equation}
In particular, if $f>0$ and its logarithm is coordinatewise
Lipschitz, with
\begin{equation}\label{eq.coordinate-lipschitz}
 |\ln f(x)-\ln f(x')|\leq L_i|x_i-x_i'|
 \quad\text{whenever $x$ and $x'$ differ only in coordinate $i$},
\end{equation}
then
\begin{equation}\label{eq.rectangle-lipschitz-kl}
 D_2(P_C\|U_C)
 \leq\frac1{2\pi^2\ln2}\sum_{i=1}^nL_i^2\ell_i^2.
\end{equation}
For any countable rectangular partition with side lengths at most
$2\eps$, a density $f$ with $\sqrt f\in H^1(\R^n)$ and its
cellwise average $\overline f$ satisfy
\begin{equation}\label{eq.partition-fisher-kl}
 D_2(f\|\overline f)
 \leq\frac{2\eps^2}{\pi^2\ln2}J(f).
\end{equation}
\end{lemma}
\begin{proof}
The sharp logarithmic Sobolev inequality for the uniform probability
law $U_I$ on an interval of length $\ell$ is
\[
 \operatorname{Ent}_{U_I}(g^2)
 \leq\frac{2\ell^2}{\pi^2}\int_I|g'|^2\,dU_I.
\]
Here $\operatorname{Ent}_\nu(q)=\int q\ln q\,d\nu-
(\int q\,d\nu)\ln\int q\,d\nu$ uses natural logarithms.
The inequality follows by scaling
\cite[Theorem~2.2]{GhangMartinWaruhiu2014}; its Sobolev formulation
follows by approximation. Standard tensorization of entropy
\cite{Boucheron2013} gives, for the product law $U_C$,
\[
 \operatorname{Ent}_{U_C}(g^2)
 \leq\frac2{\pi^2}\sum_{i=1}^n\ell_i^2
                  \int_C|\partial_i g|^2\,dU_C.
\]
Apply this inequality to $g=\sqrt{|C|f/p_C}$. Its square has
$U_C$-integral one, its entropy is $D(P_C\|U_C)$ in nats,
and
\[
 \int_C|\partial_i g|^2\,dU_C
 =\frac1{p_C}\int_C|\partial_i\sqrt f|^2\,dx.
\]
Dividing by $\ln2$ proves \eqref{eq.rectangle-kl}. These are
derivatives inside $C$; no derivative of a boundary indicator is
taken. Under \eqref{eq.coordinate-lipschitz}, weak differentiation
gives $4|\partial_i\sqrt f|^2\leq L_i^2f$ almost everywhere,
which proves \eqref{eq.rectangle-lipschitz-kl}.

For the last assertion, denote the rectangles of positive probability
by $C_s$, their probabilities by $p_s$, and their side lengths by
$\ell_i(s)$. Relative entropy decomposes over the rectangles as
\[
 D_2(f\|\overline f)=\sum_s p_sD_2(P_{C_s}\|U_{C_s}).
\]
Apply \eqref{eq.rectangle-kl} to each rectangle and use
$\ell_i(s)\leq2\eps$. All gradient integrals are nonnegative,
so countable additivity gives
\[
 \begin{split}
 D_2(f\|\overline f)
 &\leq\frac2{\pi^2\ln2}
       \sum_s\sum_{i=1}^n\ell_i(s)^2
                         \int_{C_s}|\partial_i\sqrt f|^2\,dx\leq\frac{8\eps^2}{\pi^2\ln2}
                  \int_{\R^n}|\nabla\sqrt f|^2\,dx
 =\frac{2\eps^2}{\pi^2\ln2}J(f).
 \end{split}
\]
Zero-probability cells contribute zero. This proves
\eqref{eq.partition-fisher-kl}.
\end{proof}

\begin{lemma}[Expected volume loss]\label{lem.volume-loss}
Suppose, for $1\leq t\leq n$, the conditional law of $X_t$ given
$X_{1:t-1}$ has a density bounded by $M$, almost surely; this includes
the unconditional density when $t=1$. Then
\begin{equation}\label{eq.volume-loss}
 \E\log_2\frac{2\eps}{\ell_t(S_{1:t})}
 \leq\frac{2M\eps\ m_\omega}{e\ln2}
 \qquad(1\leq t\leq n).
\end{equation}
where $m_\omega$ is given by \eqref{eq.momega}.
\end{lemma}
\begin{proof}
The address prefix is a function of $X_{1:t-1}$. Conditioning on that
prefix averages the conditional source densities, so the resulting
density of $X_t$ is also bounded by $M$. For any fixed prefix, the
probability of a scalar cell of length $\ell$ is at most $M\ell$.
Only short cells contribute a nonzero term. By
Lemma~\ref{lem.width-structure} there are at most $m_\omega$ of them,
without any separation assumption. For $0<\ell\leq2\eps$, maximizing $u\ln(1/u)$ on $(0,1]$ gives
\[
 \ell\log_2\frac{2\eps}{\ell}
 \leq\frac{2\eps}{e\ln2}.
\]
 Multiplying by $M$,
summing over the short cells, and averaging over prefixes proves
\eqref{eq.volume-loss}.
\end{proof}

\subsection{Certified approximation to the process $\eps$-entropy}\label{sec.grid-certificates}
Uniform-grid initialization gives particularly strong certificates.
For every $\eps>0$ and every origin $\theta\in\R$, initialize any
$\omega$ distinct representatives in $\theta+2\eps\Z$.
With the consistent boundary ownership defined in
Section~\ref{sec.fixed-grid}, MRR realizes that uniform grid, and its
addresses and the grid indices have identical block entropy at every
length (Proposition~\ref{prop.lattice-bijection}). The original boundary
rule gives the same conclusion whenever the source assigns zero probability
to grid boundaries. The origin is arbitrary, and the spacing $2\eps$
can be any positive real number.

Two sharper comparisons follow. At each fixed lattice state, the scalar
address entropy is at most one bit above the scalar $\varepsilon$-entropy, with the same
boundary qualification (Corollary~\ref{cor.lattice-one-bit}), improving
the two-bit comparison for general separated geometries in
Theorem~\ref{thm.two-bit}. For stationary sources, the one-bit process
comparison combines with the Fisher-information estimate to give
\begin{equation}\label{eq.grid-summary-rate}
 0\leq R_{\mathrm{MRR}}(X;A_0)-R_\eps(X)
 \leq\min\left\{1,\frac{2\mathcal J\eps^2}{\pi^2\ln2}\right\}.
\end{equation}
The quadratic term requires the Fisher-information and entropy
 assumptions of Theorem~\ref{thm.near-floor} below. It follows from
Proposition~\ref{prop.lattice} because every grid cell has full width,
so the volume penalty vanishes. If the block log densities have uniformly
bounded coordinate scores $L$, then $\mathcal J\leq L^2$.
These bounds are uniform over the grid origin and representative capacity;
for a fixed regular source, the rate gap tends to zero as
$\eps\downarrow0$ with no restriction on the growth of $\omega$.
The theorem below gives the general-initialization bound; the grid
equivalence and its sharper consequences are proved in
Section~\ref{sec.fixed-grid}.

\begin{theorem}[Quantitative process comparison]\label{thm.near-floor}
Let $X$ be stationary and assume the following:
\begin{enumerate}[label=(\roman*),leftmargin=*]
\item for every $n\geq1$, $X_{1:n}$ has a density $f_n$ with
$\int f_n|\log_2f_n|<\infty$ and
$\sqrt{f_n}\in H^1(\R^n)$;
\item the asymptotic Fisher information per
coordinate satisfies
\begin{equation}\label{eq.fisher-rate}
 \mathcal J=\limsup_{n\to\infty}\frac{J(f_n)}{n}<\infty
\end{equation}
\item $H(\floor{X_1})<\infty$ and the differential entropy rate
$\overline h(X)=\lim_{n\to\infty}h(X_{1:n})/n$ is a finite real number.
\end{enumerate}
For any tolerance $\eps>0$, capacity $\omega\geq1$, and fixed initial
state $A_0\in\A_\omega$, define
\begin{equation}\label{eq.delta}
 \Delta_{\omega,\eps}(\mathcal J)
 =\frac{2\mathcal J\eps^2}{\pi^2\ln2}
  +\frac{m_\omega\eps\sqrt{\mathcal J}}{e\ln2},
 \qquad m_\omega=\floor{\frac{3\omega}{2}}-1.
\end{equation}
Denoting $J_n=J(f_n)$, at every block length we have
\begin{equation}\label{eq.finite-block-bracket}
 \begin{split}
 h(X_{1:n})-n\log_2(2\eps)
 \leq C_{n,\eps}(X)\leq H(S_{1:n})
 &\leq h(X_{1:n})-n\log_2(2\eps)
       +\frac{2\eps^2J_n}{\pi^2\ln2}
       +\frac{m_\omega\eps\sqrt{nJ_n}}{e\ln2}.
 \end{split}
\end{equation}
In particular,
\begin{equation}\label{eq.process-bracket}
 \overline h(X)-\log_2(2\eps)
 \leq R_\eps(X)
 \leq R_{\mathrm{MRR}}(X;A_0)
 \leq\overline h(X)-\log_2(2\eps)
       +\Delta_{\omega,\eps}(\mathcal J).
\end{equation}
Thus $0\leq R_{\mathrm{MRR}}-R_\eps\leq\Delta_{\omega,\eps}(\mathcal J)$.
If the initial geometry is $\eps$-separated, this upper bound on the
gap can be replaced by $\min\{2,\Delta_{\omega,\eps}(\mathcal J)\}$.
For a fixed source satisfying these assumptions, the gap tends to zero
as $\eps\downarrow0$ whenever $\omega(\eps)\eps\to0$.

If, additionally, the conditional density of $X_t$ given $X_{1:t-1}$
is bounded by  $M$ almost surely, for every $t\geq1$
 the last term in
\eqref{eq.finite-block-bracket} may be replaced by
\[
 \frac{m_\omega\eps}{e\ln2}
       \min\{\sqrt{nJ_n},\,2nM\}
\]
and all rate conclusions hold with
$\Delta_{\omega,\eps}(\mathcal J)$ replaced by
\[
 \Delta^{(M)}_{\omega,\eps}(\mathcal J)
 =\frac{2\mathcal J\eps^2}{\pi^2\ln2}
  +\frac{m_\omega\eps}{e\ln2}
       \min\{\sqrt{\mathcal J},\,2M\}.
\]
\end{theorem}

\begin{proof}
We first check that the discrete entropies are finite.
For any geometry, a unit interval intersects at most
\[
 K_{\omega,\eps}=2\omega+2+\ceil{1/(2\eps)}
\]
address cells. Indeed, at most $m_\omega$ cells are short by
Lemma~\ref{lem.width-structure}. Among the full-width cells meeting
the interval, at most two are not contained in it, and the contained
cells have disjoint interiors of length $2\eps$. The displayed
number is therefore a uniform upper bound. Conditional on an address
prefix and $\floor{X_t}$, the next address has at most
$K_{\omega,\eps}$ possibilities. The entropy chain rule gives
\[
 H(S_t\mid S_{1:t-1})
 \leq H(\floor{X_t})+\log_2K_{\omega,\eps}
 =H(\floor{X_1})+\log_2K_{\omega,\eps}<\infty.
\]
Consequently $H(S_{1:n})<\infty$ and $C_{n,\eps}<\infty$.

We next bound the expected volume loss using only $J_n=J(f_n)$.
For any probability density $q$ on $\R$ with $\sqrt q\in H^1(\R)$,
the function $q$ belongs to $W^{1,1}(\R)$ and its absolutely
continuous representative vanishes at both infinities. Hence
\[
 2\|q\|_\infty
 \leq\int_\R |q'(x)|\,dx
 \leq 2\|\sqrt q\|_2\|(\sqrt q)'\|_2
 =\sqrt{J(q)}.
\]
Fix $n$ and $1\leq t\leq n$, and write
\[
 X_{-t}^{(n)}=(X_1,\ldots,X_{t-1},X_{t+1},\ldots,X_n).
\]
Let $q_{n,t}$ denote the conditional density of $X_t$ given
$X_{-t}^{(n)}$. Sobolev slicing of $\sqrt{f_n}$ shows that
$\sqrt{q_{n,t}}\in H^1(\R)$ almost surely. Since the normalizing
marginal density does not depend on the $t$th coordinate,
\[
 \E J(q_{n,t})
 =4\int_{\R^n}|\partial_t\sqrt{f_n(x)}|^2\,dx
 =:J_{n,t},
 \qquad \sum_{t=1}^n J_{n,t}=J_n.
\]
Conditioning on $X_{-t}^{(n)}$ fixes the entire input past and hence
the scalar address partition at time $t$. Each cell of length $\ell$
has conditional probability at most $\ell\|q_{n,t}\|_\infty$.
Only the at most $m_\omega$ short cells contribute to the volume
loss, and
\[
 \sup_{0<\ell\leq2\eps}\ell\log_2\frac{2\eps}{\ell}
 =\frac{2\eps}{e\ln2}.
\]
It follows that
\[
 \E\!\left[
  \log_2\frac{2\eps}{\ell_t(S_{1:t})}
  \,\middle|\,X_{-t}^{(n)}\right]
 \leq\frac{2m_\omega\eps}{e\ln2}\|q_{n,t}\|_\infty
 \leq\frac{m_\omega\eps}{e\ln2}\sqrt{J(q_{n,t})}.
\]
Taking expectations, summing, and applying Jensen's and
Cauchy--Schwarz inequalities yields
\[
 \begin{split}
 V_n&:=\sum_{t=1}^n\E\log_2\frac{2\eps}{\ell_t(S_{1:t})}\leq\frac{m_\omega\eps}{e\ln2}
             \sum_{t=1}^n\sqrt{J_{n,t}}
 \leq\frac{m_\omega\eps\sqrt{nJ_n}}{e\ln2}.
 \end{split}
\]
In particular, since
\[
 \log_2v(S_{1:n})
 =n\log_2(2\eps)-\sum_{t=1}^n
                         \log_2\frac{2\eps}{\ell_t(S_{1:t})},
\]
the random variable $\log_2v(S_{1:n})$ is integrable, as required in
Proposition~\ref{prop.rectangle-identity}.

The address preimages in \eqref{eq.adaptive-rectangle} are
rectangles with every side length at most $2\eps$. Their cellwise
density average therefore satisfies, by
Lemma~\ref{lem.rectangle-kl},
\[
 D_2(f_n\|\overline f_n)
 \leq\frac{2\eps^2}{\pi^2\ln2}J_n.
\]
Substituting this estimate and the bound on $V_n$ into
\eqref{eq.entropy-decomposition} proves the upper bound in
\eqref{eq.finite-block-bracket}. The other bounds follow from
\eqref{eq.shannon-lower} and the admissibility of the address partition.

For a stationary source the limit in \eqref{eq.process-floor} exists,
and the differential entropy rate exists by assumption. Divide
\eqref{eq.finite-block-bracket} by $n$ and take limits, using
\eqref{eq.fisher-rate} and the limsup definition of
$R_{\mathrm{MRR}}$. This proves \eqref{eq.process-bracket}.
Subtracting the first bound from the last proves the gap estimate.
Both terms of \eqref{eq.delta} tend to zero under the stated
capacity--tolerance scaling. For separated initialization,
Corollary~\ref{cor.two-bit-process} supplies the additional two-bit
upper bound.

Under the additional conditional-density bound,
Lemma~\ref{lem.volume-loss} also gives
$V_n\leq2nM m_\omega\eps/(e\ln2)$. Taking the minimum of the
two bounds on $V_n$ proves the finite-block refinement. Dividing by
$n$ and taking the limsup gives the stated rate refinement.
\end{proof}

When the block log densities obey
\eqref{eq.coordinate-lipschitz} with the same constant $L$ for all
coordinates and block lengths, $J_n\leq nL^2$ and hence
$\mathcal J\leq L^2$. The Fisher formulation also applies when
scores are unbounded, including nondegenerate Gaussian laws, provided
their squared scores have the stated integral bound.

If $\overline h(X)$ and an upper bound on $\mathcal J$ are known,
the outer bounds in \eqref{eq.process-bracket} give an explicit
certified interval for $R_\eps(X)$. Its lower endpoint may be replaced
by its maximum with zero. A known conditional-density bound $M$
may sharpen the upper endpoint through the optional refinement.
The certificate does not require the representative state to be
sufficient for prediction.
\begin{corollary}[Autoregressive Laplace source]\label{cor.ar-laplace}
Let $(X_t)_{t\in\Z}$ be the stationary solution of
\begin{equation}\label{eq.ar-source}
 X_t=aX_{t-1}+Z_t,\qquad |a|<1,
 \qquad f_Z(z)=\frac1{2b}e^{-|z|/b},\quad b>0,
\end{equation}
where the innovations are independent. Put $r=\eps/b$. The process $\eps$-entropy satisfies
\begin{equation}\label{eq.ar-bound}
 \max\left\{0,\log_2\frac e r\right\}
 \leq R_\eps(X)
 \leq\log_2\frac e r+
             \frac{2(1+a^2)r^2}{\pi^2\ln2}.
\end{equation}
For every capacity $\omega\geq1$, MRR initialized on a lattice of
spacing $2\eps$ satisfies
\begin{equation}\label{eq.ar-lattice-rate}
 R_\eps(X)\leq R_{\mathrm{MRR}}(X;A_0)
 \leq\log_2\frac e r+
             \frac{2(1+a^2)r^2}{\pi^2\ln2}.
\end{equation}
For an arbitrary fixed initial state, the rate guarantee is
\begin{equation}\label{eq.ar-general-rate}
 \begin{split}
 R_\eps(X)&\leq R_{\mathrm{MRR}}(X;A_0)\leq\log_2\frac e r
   +\frac{2(1+a^2)r^2}{\pi^2\ln2}
   +\frac{r}{e\ln2}\left(\left\lfloor\frac{3\omega}{2}\right\rfloor-1\right).
 \end{split}
\end{equation}
\end{corollary}
\begin{proof}
The stationary solution is the convergent linear process
$X_t=\sum_{j\geq0}a^jZ_{t-j}$; existence and stationarity follow
from the standard autoregressive construction
\cite{BrockwellDavis1991}. It has a finite first absolute moment.
Its marginal density can be written
$\pi(x)=\E f_Z(x-aX_0)$, so it is positive, bounded by $1/(2b)$,
and satisfies
\[
 e^{-|x-y|/b}\pi(y)\leq\pi(x)
                         \leq e^{|x-y|/b}\pi(y).
\]
The Markov factorization is
\begin{equation}\label{eq.ar-density-factorization}
 f_n(x_1,\ldots,x_n)=\pi(x_1)
                         \prod_{t=2}^n f_Z(x_t-ax_{t-1}).
\end{equation}
Every one-step conditional density, including the first marginal,
is bounded by $M=1/(2b)$.

Let $\rho(z)=f_Z'(z)/f_Z(z)=-\operatorname{sgn}(z)/b$ almost
everywhere, and let $\rho_\pi=\pi'/\pi$ be the marginal score.
Differentiating the convolution in the weak sense gives
\[
 \rho_\pi(X_1)=\E[\rho(Z_1)\mid X_1].
\]
Conditional Jensen therefore gives
$J(\pi)=\E\rho_\pi(X_1)^2\leq\E\rho(Z_1)^2=1/b^2$.
In particular, $\sqrt\pi\in H^1(\R)$. Also
$\E\rho(Z_1)=0$ by symmetry.

For $n\geq2$, differentiation of
\eqref{eq.ar-density-factorization} gives the block scores
\[
 \begin{aligned}
 \partial_1\ln f_n(X_{1:n})
   &=\rho_\pi(X_1)-a\rho(Z_2),\\
 \partial_i\ln f_n(X_{1:n})
   &=\rho(Z_i)-a\rho(Z_{i+1})\quad(2\leq i\leq n-1),\\
 \partial_n\ln f_n(X_{1:n})&=\rho(Z_n).
 \end{aligned}
\]
The first cross term vanishes because $X_1$ and $Z_2$ are
independent and $\rho(Z_2)$ has mean zero. All interior cross terms
vanish because distinct innovations are independent and their scores
have mean zero. Squaring, taking expectations, and adding the
coordinate contributions therefore yields
\begin{equation}\label{eq.ar-laplace-fisher}
 J_n=J(\pi)+\frac{(n-1)(1+a^2)}{b^2},\qquad n\geq1,
 \qquad \mathcal J=\frac{1+a^2}{b^2}.
\end{equation}
For $n=1$ the identity is immediate. The positive block density is
locally absolutely continuous, and these integrable squared scores
also verify $\sqrt{f_n}\in H^1(\R^n)$.

The finite first absolute moment implies
$H(\floor{X_1})<\infty$ by comparison with a two-sided geometric
mass function and the entropy variational inequality
\cite{CoverThomas2006}. The ratio inequality for $\pi$ bounds
$|\ln\pi(x)|$ by a constant plus $|x|/b$. Together with
\eqref{eq.ar-density-factorization} and the finite first moments,
this proves $\int f_n|\ln f_n|<\infty$. The differential-entropy
chain rule gives
\[
 h(X_{1:n})=h(X_1)+(n-1)h(Z_1),\qquad
 \overline h(X)=h(Z_1)=\log_2(2eb).
\]
 The bounded-density refinement of Theorem~\ref{thm.near-floor},
 with $M=1/(2b)$, proves \eqref{eq.ar-general-rate}.
 Proposition~\ref{prop.lattice}
eliminates the volume term for lattice initialization and yields
\eqref{eq.ar-lattice-rate}. Such an initialization is an admissible
quantizer, so its rate bounds the process $\eps$-entropy from above.
The Shannon lower bound \eqref{eq.shannon-lower} and nonnegativity
give the lower endpoint in \eqref{eq.ar-bound}.
\end{proof}

For example, $a=0.8$ and $\eps/b=10^{-3}$ give the
outward-rounded certified bracket
\begin{equation}\label{eq.ar-numeric}
 11.408479325551\ \leq\ R_\eps(X)\ \leq\ 11.408479805007
 \quad\text{bits per sample}.
\end{equation}
The displayed interval has width less than $4.795\times10^{-7}$
bits per sample. It is a certificate for the process $\eps$-entropy regardless
of how any separately implemented MRR is initialized. A
lattice-initialized MRR has rate in the same bracket at every
capacity. For arbitrary initialization with $\omega=8$, the additional
volume term in \eqref{eq.ar-general-rate} is approximately
$0.005838116300$ bits per sample. These certificates use the specified
source parameters rather than an observed sample path.

\subsection{Invariant lattices and fixed-grid representation}\label{sec.fixed-grid}
Fix any $\theta\in\R$ and any $\eps>0$, and define
\begin{equation}\label{eq.lattice}
 \Lambda_{\theta,\eps}=\theta+2\eps\Z.
\end{equation}
The associated integer index and half-open cells are
\begin{equation}\label{eq.fixed-grid-index}
 K_{\theta,\eps}(x)=\floor{\frac{x-\theta}{2\eps}+\frac12},
 \qquad
 C_k=[\theta+(2k-1)\eps,\theta+(2k+1)\eps),\quad k\in\Z.
\end{equation}
Write $Q_\Lambda(x)=\theta+2\eps K_{\theta,\eps}(x)$ and
$\mathcal B_{\theta,\eps}=\{\theta+(2k+1)\eps:k\in\Z\}$
for the fixed set of cell boundaries.

\begin{proposition}[Invariant lattice family]\label{prop.lattice}
If $R_0\subset\Lambda_{\theta,\eps}$, every later MRR geometry is
contained in this lattice. Every nonempty address cell has length
$2\eps$, and, away from $\mathcal B_{\theta,\eps}$,
\begin{equation}\label{eq.lattice-reconstruction}
 \pi_{R_t}(x)=Q_\Lambda(x)\qquad(t\geq0).
\end{equation}
For a source having block densities, the volume term in
\eqref{eq.entropy-decomposition} is zero.  If the source satisfies the assumptions of
 Theorem~\ref{thm.near-floor}, the rate gap satisfies 
\begin{equation}\label{eq.lattice-near}
 0\leq R_{\mathrm{MRR}}-R_\eps\leq\frac{2\mathcal J\eps^2}{\pi^2\ln2}.
\end{equation}
These conclusions hold for every capacity $\omega\geq1$.
\end{proposition}
\begin{proof}
Consecutive stored lattice values are separated by $2m\eps$ for an
integer $m\geq1$. Their uncovered gap has length $2(m-1)\eps$.
The halo-edge subdivision therefore gives exactly $m-1$ full-width
cells, centered at the missing lattice values between them. If $m=1$
the gap is empty. Each proximity cell has full width $2\eps$, and
the exterior formulas give the remaining lattice cells. Thus every
refinement midpoint belongs to the lattice. The update preserves the
lattice geometry by induction. The same argument covers $\omega=1$,
when there are no interior gaps.

The cell interiors coincide with the interiors of the $C_k$, proving
\eqref{eq.lattice-reconstruction}. The proximity-priority convention
can change the owner of a boundary as the stored subset changes, but
does not change cell lengths. Every side of every address rectangle
therefore has length $2\eps$. A unit interval meets a uniformly
bounded number of these cells, so $H(\floor{X_1})<\infty$ ensures
the required address entropies are finite by the argument in
Theorem~\ref{thm.near-floor}. Also $\log_2v=n\log_2(2\eps)$ is
constant. The entropy identity therefore applies. Its volume penalty is zero, and its density-variation
term is at most $2\eps^2J_n/(\pi^2\ln2)$ by
Lemma~\ref{lem.rectangle-kl}. Divide by $n$, take the limsup,
and use \eqref{eq.fisher-rate} to obtain \eqref{eq.lattice-near}.
\end{proof}
\begin{remark}[Grid offset and entropy at scale $2\varepsilon$]
\label{rem:offset}
By Proposition~\ref{prop.lattice}, lattice initialization reproduces
the reconstruction map of a uniform quantizer of step $2\varepsilon$
whose offset $\theta$ is fixed by $R_0$. For a source with a density,
the first address of a lattice-initialized MRR is therefore, almost
surely, in bijection with the lattice index
\[
  K^{\theta}(X_1) = \Bigl\lfloor \frac{X_1-\theta}{2\varepsilon}
                    + \frac12 \Bigr\rfloor .
\]
Thus $H(S_1) = H(K^{\theta}(X_1))$ depends on $\theta$, but not on the
capacity or on which lattice values are stored. When $H(\lfloor X_1\rfloor)<\infty $, averaging over a
uniformly distributed offset gives the entropy of $X_1$ at scale
$2\varepsilon$ in the sense of Breuillard and
Varj\'u~\cite[Section~2.2]{BreuillardVarju2019}:
\[
  \frac{1}{2\varepsilon}\int_0^{2\varepsilon}
     H\bigl(K^{\theta}(X_1)\bigr)\,d\theta
  \;=\; \int_0^1 H\Bigl(\Bigl\lfloor \frac{X_1}{2\varepsilon} + t \Bigr\rfloor\Bigr)\,dt
  \;=\; h(X_1 + I) - \log_2(2\varepsilon),
\]
where $I$ is uniform on $[0,2\varepsilon]$ and independent of $X_1$.
The first equality is the change of variable
$t = \tfrac12 - \theta/(2\varepsilon)$ modulo one.
For the second, put $r=2\varepsilon$ and $W=X_1+I$. Its density is
$f_W(w)=r^{-1}\,\mathbb{P}(w-r\le X_1\le w)\le r^{-1}$. For almost every $t$,
$\mathbb{P}\bigl(\lfloor X_1/r+t\rfloor=k\bigr)=r f_W\bigl(r(k+1-t)\bigr)$ for
every $k$. The terms $-rf_W\log_2(rf_W)$ are nonnegative, so summing over $k$,
integrating over $t\in[0,1)$ and applying Tonelli's theorem gives
$-\int f_W\log_2(rf_W)\,dw=h(W)-\log_2 r$. This is the identity
of~\cite[Lemma~10]{BreuillardVarju2019}, stated there for bounded $X_1$. Both sides are finite since $H(\lfloor X_1\rfloor)<\infty$.

Every half-open cell of one offset meets at most two cells of
any other, so the overlap inequality of Section~\ref{sec:overlap}
shows that $H(K^{\theta}(X_1))$ differs from this average by at most
one bit (compare the translation estimate
in~\cite[Lemma~3.2]{Hochman2014}). The offset-averaged quantity is
therefore a natural fixed-grid benchmark for adaptive, non-lattice
geometries: it does not depend on an arbitrary choice of $\theta$,
whereas optimizing $\theta$ requires knowledge of the source.
\end{remark}
For arbitrary source laws, boundary ownership can be made consistent
with \eqref{eq.fixed-grid-index}. On a lattice geometry $R$, assign
the entire cell $C_k$ the local address of its center:
\begin{equation}\label{eq.consistent-lattice-address}
 S_R^\circ(x)=S_R\bigl(Q_\Lambda(x)\bigr).
\end{equation}
Retain the reconstruction map $d_R$ and the update map $V$. The address
is a proximity address precisely when the center belongs to $R$;
otherwise it is the refinement address centered at that lattice value.
Outside the invariant lattice family, retain the rules of
Section~\ref{sec.mrr}. This defines the consistent-boundary MRR on
the original state space. Its maps are Borel, its reconstruction error
is at most $\eps$, and it is still a causal quantizer with a shared state satisfying properties (P1)-(P4). Indeed, the modification changes only the
ownership of scalar cell endpoints. A refinement still inserts a new
lattice value at distance at least $2\eps$ from every retained value,
using the same eviction and timestamp rule.

\begin{proposition}[Causal equivalence with the fixed grid]
\label{prop.lattice-bijection}
Start from any shared deterministic $A_0$ whose $\omega$ distinct
representatives lie in $\Lambda_{\theta,\eps}$. Put
$K_t=K_{\theta,\eps}(X_t)$ and let $S_t^\circ$ be the addresses of
the consistent-boundary MRR. For every input law and every $n\geq1$,
the strings $K_{1:n}$ and $S_{1:n}^\circ$ determine each other, and
\begin{equation}\label{eq.lattice-string-entropy}
 H(S_{1:n}^\circ)=H(K_{1:n}).
\end{equation}
The entropies may be infinite. If
\begin{equation}\label{eq.lattice-null-boundaries}
 \Prob(X_t\in\mathcal B_{\theta,\eps})=0\qquad(t\geq1),
\end{equation}
the original MRR obeys the same identity with $S_{1:n}$ in place of
$S_{1:n}^\circ$.

If $X$ is stationary and $H(K_1)<\infty$, the common normalized
block entropy has a limit equal to the entropy rate
$\overline H(K)=\lim_n H(K_{1:n})/n$. This rate is independent of
the capacity, the initial stored subset, and its timestamp ordering.
In particular, under \eqref{eq.lattice-null-boundaries},
\begin{equation}\label{eq.lattice-rate-identity}
 R_{\mathrm{MRR}}(X;A_0)=\overline H(K).
\end{equation}
\end{proposition}
\begin{proof}
Given $K_{1:n}$ and $A_0$, recursively compute the address at step $t$
as $S_{R_{t-1}}(\theta+2\eps K_t)$, and then apply $V$ to obtain
$A_t$. By \eqref{eq.consistent-lattice-address}, this is precisely
the address and state generated by the consistent-boundary MRR.
Conversely, given its address string and $A_0$, synchronized decoding
recovers each reconstruction $\widehat X_t$. The reconstruction of
the selected cell is $\theta+2\eps K_t$, so
$K_t=(\widehat X_t-\theta)/(2\eps)$. These recursions give inverse
deterministic maps on the sets of realizable strings and prove
\eqref{eq.lattice-string-entropy}.

Under \eqref{eq.lattice-null-boundaries}, with probability one none
of the finitely many inputs in a specified block is a lattice boundary.
The original and consistent-boundary rules agree at the first step and,
by induction on their common state, at all subsequent steps of the
block. Thus their address strings agree almost surely.

The coordinatewise fixed map $K_{\theta,\eps}$ preserves stationarity.
The entropy-rate limit for a stationary countable-alphabet process with
finite one-letter entropy exists by the standard subadditivity argument
\cite{CoverThomas2006}. The block identities then identify the MRR
rate with this limit. Since the law of $K$ involves neither $\omega$
nor $A_0$, the asserted independence follows.
\end{proof}
\begin{remark}[Equal entropy, different representations]
Proposition~\ref{prop.lattice-bijection} identifies the address and grid-index strings
up to a causal bijection, so their block entropies and entropy rates coincide
at every capacity. It does not imply that coders with limited context or
finitely many samples perform equally on both. The retained state and the
state-relative addresses change with $\omega$ even though the grid sequence
does not, and this changes which information a simple coding model sees. At
capacity one, for instance, the address is the increment $K_t-K_{t-1}$ of the
grid index, so MRR performs differential coding of the indices. An increment
model has far fewer parameters than a model conditioned on the previous index
and can be learned from fewer samples, although conditioning is never worse in
the limit, since
$H(K_t\mid K_{t-1})=H(K_t-K_{t-1}\mid K_{t-1})\leq H(K_t-K_{t-1})$.   The entropy rate statements  here concern ideal block coding;
  finite-sample coding performance requires a separate analysis.
\end{remark}
\begin{corollary}[Sharp one-bit lattice comparison]
\label{cor.lattice-one-bit}
For every probability law $\mu$ with $F_\eps(\mu)<\infty$,
\begin{equation}\label{eq.lattice-scalar-one-bit}
 F_\eps(\mu)\leq H_\mu(K_{\theta,\eps}(X))
                 \leq F_\eps(\mu)+1.
\end{equation}
For every random vector $X_{1:n}$ with $C_{n,\eps}(X)<\infty$,
\begin{equation}\label{eq.lattice-block-one-bit}
 C_{n,\eps}(X)\leq H(K_{1:n})\leq C_{n,\eps}(X)+n.
\end{equation}
Consequently, if $X$ is stationary and $C_{1,\eps}(X)<\infty$,
\begin{equation}\label{eq.lattice-process-one-bit}
 0\leq\overline H(K)-R_\eps(X)\leq1.
\end{equation}
All three constants are sharp. At any fixed lattice state, the scalar
inequality also holds for the local address in place of the grid index.
By Proposition~\ref{prop.lattice-bijection}, the block and process
conclusions hold for consistent-boundary MRR at every capacity, and for
the original MRR under \eqref{eq.lattice-null-boundaries}. These statements
do not identify unconditional single-symbol entropies after mixing over
different states. Under the density and
entropy assumptions of Proposition~\ref{prop.lattice}, lattice initialization gives the
combined bound
\begin{equation}\label{eq.lattice-combined-gap}
 0\leq R_{\mathrm{MRR}}-R_\eps
 \leq\min\left\{1,\frac{2\mathcal J\eps^2}{\pi^2\ln2}\right\}.
\end{equation}
\end{corollary}
\begin{proof}
An interval of length at most $2\eps$ meets at most two of the
half-open cells $C_k$. To verify the endpoint case, suppose that it
contained $x\in C_j$ and $y\in C_{j+2}$. Then
$x<\theta+(2j+1)\eps$ and
$y\geq\theta+(2j+3)\eps$, implying $y-x>2\eps$, a
contradiction. Cells farther apart give the same contradiction.

Each cell of an admissible block quantizer lies in a product of $n$
closed intervals of length $2\eps$, centered at its reconstruction.
Each coordinate interval meets at most two grid cells, so the block
cell meets at most $2^n$ cells of the product grid. Inequality
\eqref{eq.partition-overlap}
therefore bounds the grid entropy by the competitor's entropy plus
$n$. Taking the infimum over competitors proves the upper bound in
\eqref{eq.lattice-block-one-bit}. The product grid itself satisfies
the error constraint, giving the lower bound. The scalar statement is
the case $n=1$. In particular, finite $C_{1,\eps}$ implies
$H(K_1)<\infty$. For a stationary source, divide the block bounds
by $n$ and use the two entropy-rate limits to obtain
\eqref{eq.lattice-process-one-bit}.

For sharpness, fix $0<\delta<\eps$ and let the inputs be independent,
each taking $\theta+\eps-\delta$ and
$\theta+\eps+\delta$ with equal probability. The grid indices are
independent fair binary variables. A constant reconstruction
$\theta+\eps$ is faithful on the source support, and its product
reconstruction is faithful on every block support. The complement can
be partitioned into additional faithful cells of probability zero.
Thus $C_{n,\eps}=0$ and $H(K_{1:n})=n$ for every $n$, proving
sharpness at all three levels. Finally,
\eqref{eq.lattice-combined-gap} combines
\eqref{eq.lattice-process-one-bit},
\eqref{eq.lattice-rate-identity}, and \eqref{eq.lattice-near}.
\end{proof}
\begin{remark}[Price of causality under the pathwise criterion]
\label{rem.price-causality}
Lattice-initialized MRR is a causal code. For a stationary source satisfying
the assumptions of~\eqref{eq.lattice-combined-gap},
\eqref{eq.causal-sandwich} therefore gives
\[
  0\le R^{\mathrm{c}}_\varepsilon(X)-R_\varepsilon(X)
  \le\min\Bigl\{1,\ \frac{2{\mathcal J}\varepsilon^2}{\pi^2\ln 2}\Bigr\},
\]
and the one-bit bound alone requires only $C_{1,\varepsilon}(X)<\infty$.
Under squared error, Linder and Zamir~\cite{LinderZamir2006} show that the
price of causality tends at high resolution to the space-filling loss of the
uniform scalar quantizer, $\tfrac12\log_2(\pi e/6)\approx0.254$ bit per
sample. That loss compares a cube with a Euclidean ball. Under the sup norm,
the error ball is itself a cube, which the product grid tiles exactly. The
only remaining gap is the density term, which vanishes at rate
$\varepsilon^2$. At a fixed tolerance the price of causality can be strictly
positive even for independent inputs; see the example in
Appendix~\ref{app.computation}.
\end{remark}

The boundary condition is essential for the one-bit comparison with
the original MRR. For example, take $\eps=1$, $\omega=1$,
$R_0=\{0\}$, and let $X$ be uniform on $\{1,2,3\}$. Its three
original MRR reconstructions are respectively $0$, $2$, and $4$,
whereas the constant reconstruction $2$ is faithful on the source
support. Hence $H(S_1)=\log_2 3$ and $F_1(\mu)=0$.

For arbitrary source laws, original MRR with lattice initialization
still has the unconditional bounds
\begin{equation}\label{eq.lattice-original-three}
 C_{n,\eps}(X)\leq H(S_{1:n})\leq C_{n,\eps}(X)+n\log_2 3.
\end{equation}
Indeed, all scalar cells have length $2\eps$; an interval of that
length can meet at most three such cells, regardless of their endpoint
ownership. Meeting four would put the two middle full-width cells
between its extreme intersections, requiring length at least $4\eps$.
Proposition~\ref{prop.causal-overlap} applies with $m=3$ and gives
\eqref{eq.lattice-original-three}. For stationary sources with finite
$C_{1,\eps}$ it also gives
$0\leq R_{\mathrm{MRR}}-R_\eps\leq\log_2 3$.

\begin{corollary}[Capacity one is a fixed grid]
\label{cor.capacity-one-grid}
Let $\omega=1$ and let the fixed initial representative be $s_0$.
Every subsequent representative belongs to
$s_0+2\eps\Z$, and every address cell has length $2\eps$.
With consistent boundary ownership, for every input law and every
$n\geq1$,
\begin{equation}\label{eq.capacity-one-grid}
 H(S_{1:n}^\circ)
 =H\bigl(K_{s_0,\eps}(X_1),\ldots,K_{s_0,\eps}(X_n)\bigr).
\end{equation}
The same identity holds for the proximity-priority MRR whenever
$\Prob(X_t\in\mathcal B_{s_0,\eps})=0$ for every $t$.
Under the regularity and entropy assumptions of
Proposition~\ref{prop.lattice}, the bound
\eqref{eq.lattice-near} therefore holds from every fixed
capacity-one initial state.
\end{corollary}
\begin{proof}
The singleton initial geometry is contained in
$\Lambda_{s_0,\eps}$. Proposition~\ref{prop.lattice} gives invariance,
full cell widths, and the rate bound. The address identity and its
boundary qualification follow from
Proposition~\ref{prop.lattice-bijection}.
\end{proof}

\subsection{Exact and near-exact floor evaluation}\label{sec.exact-floor}
\begin{theorem}[Exact floor for separated uniform cells]
\label{thm.exact-separated}\label{thm.exact-lattice}
Let $\mathcal I$ be a finite or countable index set and let
$(c_j)_{j\in\mathcal I}$ be real centers satisfying
\begin{equation}\label{eq.exact-center-separation}
 |c_i-c_j|\geq2\eps\qquad(i\ne j).
\end{equation}
Let $(J_t)_{t\in\Z}$ be stationary, with values in $\mathcal I$ and
$H(J_1)<\infty$. Let $(U_t)$ be independent, identically distributed
with law $\Unif[-\eps,\eps)$, independently of the whole process
$J$. Define
\begin{equation}\label{eq.exact-source}
 X_t=c_{J_t}+U_t.
\end{equation}
For every $n\geq1$,
\begin{equation}\label{eq.exact-block}
 C_{n,\eps}(X)=H(J_{1:n})
             =h(X_{1:n})-n\log_2(2\eps).
\end{equation}
Consequently,
\begin{equation}\label{eq.exact-rate}
 R_\eps(X)=\overline H(J)
         =\overline h(X)-\log_2(2\eps),
\end{equation}
where $\overline H(J)=\lim_n H(J_{1:n})/n$ and
$\overline h(X)=\lim_n h(X_{1:n})/n$.
MRR attains the block floor and process $\eps$-entropy  under either of the
following fixed shared initializations, with arbitrary distinct
timestamps:
\begin{enumerate}[label=(\roman*),leftmargin=*]
\item $\mathcal I$ has $m<\infty$ elements, the capacity is
$\omega=m$, and $R_0=\{c_j:j\in\mathcal I\}$;
\item all centers lie in $\Lambda_{\theta,\eps}$ for some
$\theta\in\R$,  $\omega\geq1$ and $R_0$
consists of any $\omega$ distinct values in that lattice.
\end{enumerate}
In either case,
\begin{equation}\label{eq.exact-mrr-attainment}
 H(S_{1:n})=H(J_{1:n}),\qquad
 R_{\mathrm{MRR}}(X;A_0)=R_\eps(X).
\end{equation}
In case (i), the state remains equal to $A_0$ almost surely.
\end{theorem}
\begin{proof}
Put $I_j=[c_j-\eps,c_j+\eps)$. The separation condition makes
these intervals pairwise disjoint. Conditional on $J_{1:n}=j_{1:n}$,
the input block is uniform on
$I_{j_1}\times\cdots\times I_{j_n}$, of volume $(2\eps)^n$.
Different index strings give disjoint rectangles. If
$p_{j_{1:n}}=\Prob(J_{1:n}=j_{1:n})$, the source density on that
rectangle is $p_{j_{1:n}}/(2\eps)^n$. Moreover,
$H(J_{1:n})\leq nH(J_1)<\infty$, and
\[
 \sum_{j_{1:n}:p_{j_{1:n}}>0}
 p_{j_{1:n}}\left|\log_2
          \frac{p_{j_{1:n}}}{(2\eps)^n}\right|
 \leq H(J_{1:n})+n|\log_2(2\eps)|<\infty.
\]
Thus the block differential entropy is well defined and direct
integration gives
\begin{equation}\label{eq.exact-differential}
 h(X_{1:n})=H(J_{1:n})+n\log_2(2\eps).
\end{equation}
The lower bound \eqref{eq.shannon-lower} therefore implies
$C_{n,\eps}(X)\geq H(J_{1:n})$. Conversely, reconstruct each
rectangle as $(c_{j_1},\ldots,c_{j_n})$. This is faithful throughout
the rectangle and has output entropy $H(J_{1:n})$. Intersect the
complement of their union with a fixed product grid of width
$2\eps$ to obtain additional Borel cells with faithful grid-center
reconstructions. These cells have source probability zero. This
extends the quantizer to every input block and proves
\eqref{eq.exact-block}.

For initialization (i), any input in the interior of $I_j$ is
strictly closer than $\eps$ to $c_j$ and strictly farther than
$\eps$ from every other stored center. Hence it emits the proximity
address of $c_j$ and leaves the state unchanged. The countable set
of interval endpoints has probability zero at each time; almost
surely it is never visited. Induction shows that $A_t=A_0$ for all
$t$ and that the address at each step is a fixed bijective label of
$J_t$. This proves \eqref{eq.exact-mrr-attainment} in case (i).

For initialization (ii), write $c_j=\theta+2\eps k_j$ with
distinct integers $k_j$. Almost surely,
$K_{\theta,\eps}(X_t)=k_{J_t}$, so the grid-index string and
$J_{1:n}$ determine one another. The source has a density, and
Proposition~\ref{prop.lattice-bijection} therefore gives
$H(S_{1:n})=H(J_{1:n})$. Finally, the entropy rate of a stationary
countable-alphabet process with finite one-letter entropy exists by
subadditivity \cite{CoverThomas2006}. Divide
\eqref{eq.exact-differential} and the block equalities by $n$ and
take limits. This proves \eqref{eq.exact-rate}, the existence and
finiteness of $\overline h(X)$, and the remaining rate assertions.
\end{proof}
\paragraph{Scalar entropy, block entropy, and the Shannon lower bound.}
Posner and Rodemich \cite{posner1971} distinguish the entropy
minimized over single-experiment partitions from the absolute
epsilon entropy obtained by coding independent experiments jointly.
Their diameter tolerance corresponds to $2\eps$ in the present
notation, and their natural-logarithm entropies are converted here
to bits. For a real-valued source with law $\mu$, their scalar
quantity corresponds to $F_\eps(\mu)$, whereas their absolute
quantity corresponds to $R_\eps$ for the independent process with
marginal law $\mu$. Their coding results give the latter quantity
an operational interpretation as the asymptotic information rate
required under the prescribed fidelity constraint.

In particular, the radial formulation identifies the independent-source
block $\eps$-entropy with the mutual-information infimum
\begin{equation}\label{eq.posner-radial-floor}
 R_\eps
 =
 \inf_{\substack{P_{Y\mid X}\\
                 \Prob(|X-Y|\leq\eps)=1}}
 I(X;Y),
 \qquad X\sim\mu,
\end{equation}
where reconstruction values range over $\R$.
This is the zero-distortion value for the violation distortion
$d_\eps(x,y)=\mathbf{1}_{\{|x-y|>\eps\}}$.
It is an identification with a rate-distortion optimization,
rather than with its Shannon lower bound.
For a source having a density and the requisite finite entropies,
the relevant inequalities are
\begin{equation}\label{eq.scalar-block-shannon-order}
 h(X)-\log_2(2\eps)
 \leq R_\eps
 \leq F_\eps(\mu).
\end{equation}
Neither inequality is an equality in general.
The scalar--block comparison of Posner and Rodemich
\cite[Theorem~2]{posner1971} yields relative agreement as the
entropy becomes large; it does not assert exact equality at a
fixed positive tolerance.

\begin{corollary}[Equality with the Shannon lower bound]
\label{cor.uniform-cell-shannon-equality}
Let $(c_j)_{j\in\mathcal I}$ satisfy
\eqref{eq.exact-center-separation}. Let $J$ have values in
$\mathcal I$ with $H(J)<\infty$, and let $U$ be independent of $J$
and uniform on $[-\eps,\eps)$. If
\[
 X=c_J+U,
 \qquad \mu=\mathcal L(X),
\]
then
\begin{equation}\label{eq.scalar-uniform-cell-equality}
 F_\eps(\mu)
 =H(J)
 =h(X)-\log_2(2\eps).
\end{equation}
For the independent process with marginal law $\mu$, these
quantities also equal $R_\eps$.

More generally, for the stationary process in
Theorem~\ref{thm.exact-separated}, with marginal law
$\mu=\mathcal L(X_1)$,
\begin{equation}\label{eq.stationary-uniform-cell-scalar}
 F_\eps(\mu)
 =H(J_1)
 =h(X_1)-\log_2(2\eps),
\end{equation}
and
\begin{equation}\label{eq.stationary-uniform-cell-rate}
 R_\eps
 =\overline H(J)
 =\overline h(X)-\log_2(2\eps),
\end{equation}
with $R_{\mathrm{MRR}}=R_\eps$ under either initialization in
Theorem~\ref{thm.exact-separated}. Finite, possibly non-lattice
centers may all be stored; for centers on a common lattice, every
positive capacity is sufficient with lattice initialization.
\end{corollary}

\begin{proof}
Apply Theorem~\ref{thm.exact-separated} to independent copies of
$J$ and $U$. At $n=1$, its block equality reads
\[
 F_\eps(\mu)=C_{1,\eps}(X)=H(J)
 =h(X)-\log_2(2\eps).
\]
For independent indices, $H(J_{1:n})=nH(J)$, so the same theorem
gives $R_\eps=H(J)$. Applied to an arbitrary stationary $J$, its
scalar case gives \eqref{eq.stationary-uniform-cell-scalar}, and
\eqref{eq.exact-rate} gives
\eqref{eq.stationary-uniform-cell-rate}. Its MRR attainment
statement supplies the claimed initialization conditions.
\end{proof}
For dependent center indices, the scalar $\varepsilon$-entropy involves
$H(J_1)$, while the process $\eps$-entropy involves $\overline H(J)$;
these quantities need not coincide. The equality with the
respective Shannon lower bounds follows from the uniform
within-cell structure at the specified tolerance.
For example, if $J$ is a stationary finite-state Markov chain with
stationary probabilities $\rho_i$ and transition probabilities $P_{ij}$,
the standard Markov entropy-rate formula \cite{CoverThomas2006} gives
the exact value
\begin{equation}\label{eq.exact-markov}
 R_\eps(X)=-\sum_{i,j}\rho_iP_{ij}\log_2P_{ij},
\end{equation}
with $0\log_2 0=0$. The exact-attainment class
\eqref{eq.exact-source} is specified at the chosen tolerance.
\begin{remark}[Histograms]\label{rem.histograms}
Suppose $\mu$ has a piecewise uniform density with bin width $2\varepsilon$: for some
$\theta\in\R$ and probabilities $(p_k)_{k\in\Z}$ with bin entropy
$H(K)=-\sum_k p_k\log_2 p_k<\infty$,
\[
  f(x)=\frac{p_k}{2\varepsilon},\qquad
  x\in[\theta+(2k-1)\varepsilon,\ \theta+(2k+1)\varepsilon).
\]
Then $X=\theta+2\varepsilon K+U$, where $K$ has law $(p_k)$ and $U$ is uniform
on $[-\varepsilon,\varepsilon)$ and independent of $K$. This is the lattice
case of Theorem~\ref{thm.exact-separated}, so
Corollary~\ref{cor.uniform-cell-shannon-equality} gives
\[
  F_\varepsilon(\mu)=H(K)=h(X)-\log_2(2\varepsilon).
\]
Thus, at tolerance $\varepsilon$, every histogram with bin width $2\varepsilon$
and finite bin entropy has scalar $\varepsilon$-entropy equal to its Shannon
lower bound. For independent samples this is also the process $\eps$-entropy. By case
(ii) of Theorem~\ref{thm.exact-separated}, MRR initialized with any
$\omega\ge1$ distinct points of $\theta+2\varepsilon\Z$ attains it at every
block length. Under the independent within-bin noise assumptions of
Theorem~\ref{thm.exact-separated}, the same initialization
also attains the process $\eps$-entropy for stationary dependent
bin indices.
\end{remark}
\begin{proposition}[A non-lattice absorbing geometry beats every grid offset]
\label{prop.capacity-advantage}
Fix $\eps>0$ and an integer $m\ge2$, and set
\[
 c_r=\left(4+\frac{2}{m}\right)r\eps,
 \qquad 0\le r<m.
\]
Let $(J_t)$ be stationary on $\{0,\ldots,m-1\}$ with uniform
marginal distribution. Let $(U_t)$ be independent with common
law $\Unif[-\eps,\eps)$, and independent of the entire process
$(J_t)$. Define $X_t=c_{J_t}+U_t$, so that its marginal law is
\begin{equation}\label{eq.two-interval-source}
 \mu=\frac1m\sum_{r=0}^{m-1}
       \Unif[c_r-\eps,c_r+\eps).
\end{equation}
Write
\[
 \overline H(J)=\lim_{n\to\infty}\frac1nH(J_{1:n}),
 \qquad
 \Delta_m=\frac1m\sum_{r=0}^{m-1}h_2(r/m),
\]
where $h_2(p)=-p\log_2p-(1-p)\log_2(1-p)$, with
$h_2(0)=h_2(1)=0$.

Capacity-$m$ MRR initialized with
$R_0=\{c_0,\ldots,c_{m-1}\}$ and any distinct timestamps
satisfies, for every $n\ge1$,
\begin{equation}\label{eq.two-interval-exact}
 \begin{gathered}
 C_{n,\eps}(X)=H(S_{1:n})=H(J_{1:n}),\qquad
 F_\eps(\mu)=\log_2m,\\
 R_\eps(X)=R_{\mathrm{MRR}}(X;A_0)=\overline H(J).
 \end{gathered}
\end{equation}
For the half-open fixed grids of width $2\eps$ defined in
\eqref{eq.fixed-grid-index},
\begin{equation}\label{eq.two-interval-grid}
 \inf_{\theta\in\R}
 H\bigl(K_{\theta,\eps}(X_1),\ldots,
        K_{\theta,\eps}(X_n)\bigr)
 =H(J_{1:n})+n\Delta_m.
\end{equation}
A minimizing grid offset can be chosen independently of $n$.
Consequently, with the infima taken over fixed shared initial
states,
\begin{equation}\label{eq.capacity-one-two}
 \begin{aligned}
 \inf_{A_0\in\A_1}R_{\mathrm{MRR}}(X;A_0)
   &=\overline H(J)+\Delta_m,\\
 \inf_{A_0\in\A_m}R_{\mathrm{MRR}}(X;A_0)
   &=\overline H(J).
 \end{aligned}
\end{equation}
For $m=2$, the centers are $0$ and $5\eps$ and
$\Delta_2=1/2$. If the band indices are also independent,
then $\overline H(J)=\log_2m$, recovering rates $3/2$ and $1$
in the two-band case.
All rates are in bits per sample, and the grid comparison
permits optimal block entropy coding of the grid indices.
\end{proposition}

\begin{proof}
Put $B_r=[c_r-\eps,c_r+\eps)$.
Since the center spacing exceeds $4\eps$, every observation
in $B_r$ produces a proximity hit at $c_r$.
No refinement occurs, the stored state remains fixed, and
the addresses identify the band indices. Hence
\[
 H(S_{1:n})=H(J_{1:n}).
\]
Conversely, an admissible reconstruction of a point in $B_r$
belongs to $[c_r-2\eps,c_r+2\eps]$.
These intervals are pairwise disjoint, so every admissible
block label determines $J_{1:n}$. Its entropy is therefore
at least $H(J_{1:n})$. Reconstruction at the band centers
attains this bound. Outside the bands, extend the quantizer
using a width-$2\eps$ grid with separate labels; these
additional labels have probability zero.
This proves the block identity in
\eqref{eq.two-interval-exact}, and its scalar specialization
gives $F_\eps(\mu)=\log_2m$.

To optimize over grids, parameterize them by their boundary
offset $a\in\R$ and write
\[
 K_{a,t}=\left\lfloor\frac{X_t-a}{2\eps}\right\rfloor.
\]
This ranges over the same family of half-open grids as
\eqref{eq.fixed-grid-index}.
The distance between consecutive bands exceeds $2\eps$,
so a grid cell cannot meet two bands. Thus $K_{a,1:n}$
determines $J_{1:n}$.

Set
\[
 b_r=c_r-\eps,\qquad
 p_r(a)=\left\{\frac{b_r-a}{2\eps}\right\},
\]
where braces denote fractional part.
The band $B_r$ meets at most two grid cells, with
probabilities $p_r(a)$ and $1-p_r(a)$, up to null boundary
events. Conditional on $J_{1:n}$, the grid outputs are
independent because the noises $U_t$ are independent.
The chain rule and the uniform band marginals give
\[
 H(K_{a,1:n})
 =H(J_{1:n})+\frac{n}{m}
   \sum_{r=0}^{m-1}h_2(p_r(a)).
\]

Since
\[
 \frac{b_r}{2\eps}=2r+\frac rm-\frac12,
\]
the phases are equally spaced modulo one.
Their unordered set can be written as
$\{u,u+1/m,\ldots,u+(m-1)/m\}$ for some $0\le u<1/m$.
The function
\[
 F_m(u)=\frac1m\sum_{r=0}^{m-1}h_2(u+r/m),
 \qquad 0\le u\le1/m,
\]
is concave and satisfies
\[
 F_m(0)=F_m(1/m)=\Delta_m.
\]
Therefore $F_m(u)\ge\Delta_m$, with equality when one band
aligns with a grid cell. This proves
\eqref{eq.two-interval-grid}, with an optimizing offset
independent of $n$.

Stationarity and the finite alphabet ensure that
$\overline H(J)$ exists. Dividing the block identities by
$n$ and taking limits gives the rate identities.
Every capacity-one initialization realizes a uniform grid,
apart from boundary conventions that are immaterial for
these absolutely continuous block laws; every grid origin
is available through such an initialization.
This proves the first equality in
\eqref{eq.capacity-one-two}.
The second follows from the universal lower bound
$R_{\mathrm{MRR}}\ge R_\eps$ and attainment by the
specified capacity-$m$ state.
\end{proof}
In the two-band case $m=2$ of
Proposition~\ref{prop.capacity-advantage},
the centers $0$ and $5\eps$ do not lie in any common lattice of
spacing $2\eps$. Proposition~\ref{prop.capacity-advantage} thus
distinguishes non-lattice reconstruction geometry from a causal
relabelling of a fixed grid. The comparison holds even after
optimizing the fixed-grid offset and allowing unrestricted block
entropy coding. It uses a source-matched, absorbing initialization:
no representative moves on the source path. It establishes a
benefit of the stored geometry, without implying convergence to
that geometry from arbitrary initialization or monotonic rate
improvement for general sources as capacity increases.
\begin{corollary}[Optimized-grid penalty for separated uniform bands]
\label{cor:sharp-uniform-band-penalty}
Let $\mathcal B_\varepsilon$ be the class
of stationary sources
\[
X_t=c_{J_t}+U_t,
\]
where $(J_t)$ takes values in a finite set,
the distinct centers satisfy
$|c_r-c_s|>4\varepsilon$, and $(U_t)$ is i.i.d.,
uniform on $[-\varepsilon,\varepsilon)$, and independent
of the process $J$.
For $X\in\mathcal B_\varepsilon$, define
\[
\Gamma_\varepsilon(X)
=\inf_{a\in\mathbb R}R_{\mathrm{grid}}(a)-R_\varepsilon(X)\qquad R_{\mathrm{grid}}(a)=\mathop{\lim}_{n\to\infty} \frac{H(K_{a,1:n})}{n}
\]
where $R_{\mathrm{grid}}(a)$ is the entropy rate of the grid indices at offset $a$. Then
\[
0\le\Gamma_\varepsilon(X)\le\frac1{2\ln2},
\qquad
\sup_{X\in\mathcal B_\varepsilon}\Gamma_\varepsilon(X)
=\frac1{2\ln2}.
\]
In particular, the sources of
Proposition~\ref{prop.capacity-advantage} satisfy
\[
\Gamma_\varepsilon(X)=\Delta_m,
\qquad
\lim_{m\to\infty}\Delta_m
=\frac1{2\ln2}=0.7213475204\ldots.
\]
For each source in this class, MRR initialized at 
its band centers, with capacity equal to their number, attains the optimal rate $R_\varepsilon(X)$.
\end{corollary}

\begin{proof}
Write $\pi_r=\Pr(J_t=r)$. The separation argument in
Proposition~\ref{prop.capacity-advantage} gives
$C_{n,\varepsilon}(X)=H(J_{1:n})$ and, for every offset $a$,
\[
H(K_{a,1:n})=H(J_{1:n})+nG(a),
\]
where
\[
G(a)=\sum_r\pi_r
h_2\!\left(
\left\{\frac{c_r-\varepsilon-a}{2\varepsilon}\right\}
\right).
\]
Consequently,
$\Gamma_\varepsilon(X)=\inf_aG(a)\ge0$.
Averaging over one period of the grid offset gives
\[
\frac1{2\varepsilon}\int_0^{2\varepsilon}G(a)\,da
=\sum_r\pi_r\int_0^1h_2(u)\,du
=\frac1{2\ln2}.
\]
Thus $\inf_aG(a)\le1/(2\ln2)$.

For the equally weighted sources of
Proposition~\ref{prop.capacity-advantage},
the exact penalty is $\Delta_m$. Since $h_2$ is continuous
on $[0,1]$, its Riemann sums satisfy
\[
\lim_{m\to\infty}\Delta_m
=\int_0^1h_2(u)\,du
=\frac1{2\ln2}.
\]
This proves sharpness of the universal bound.
The MRR attainment statement follows because initialization
at all band centers produces only proximity hits and
preserves the band-index sequence.
\end{proof}
\begin{remark}[An isolated absorbing geometry]
\label{rem.isolated-geometry}
For the two-band case $m=2$ of
Proposition~\ref{prop.capacity-advantage}, the unordered
capacity-two geometry $\{0,5\eps\}$ is the unique absorbing
geometry.
It corresponds to two age orderings and to any assignment of distinct
absolute timestamps. It is also a strict local minimizer of the
one-step address entropy: if $r_1<r_2$,
$|r_1|<\eps/2$, and $|r_2-5\eps|<\eps/2$, then
\begin{equation}\label{eq.isolated-instantaneous-minimum}
 H_\mu(S_{\{r_1,r_2\}}(X))=1
 \quad\Longleftrightarrow\quad
 (r_1,r_2)=(0,5\eps).
\end{equation}
Thus the exact absorbing geometry is isolated among geometries attaining
the minimum one-step entropy locally. It is not an isolated point of
the ambient state space, and no isolation claim is made for minimizers
of the asymptotic rate over initial states.
\end{remark}
\begin{proof}
An absorbing geometry must cover both source intervals, up to null
sets, by its proximity halos. A halo has length $2\eps$ and cannot
meet both intervals with positive probability, since their separating
gap has length $3\eps$. With only two halos, one must therefore cover
each entire source interval. Both intervals have length $2\eps$,
forcing the halo centers to be $0$ and $5\eps$. Conversely, these
centers make every observation a proximity event almost surely.
For the local assertion, every address cell has diameter at most
$2\eps$, so the address determines the mixture component $J$.
Consequently its entropy is
$1+\tfrac12H(S\mid J=1)+\tfrac12H(S\mid J=2)$.
In the stated neighborhood each representative lies inside its own
source interval and the other representative's halo does not meet
that interval. If a representative is displaced from the interval's
center, its proximity cell captures a positive proper fraction of
that interval; the remaining positive fraction emits other addresses.
The corresponding conditional entropy is strictly positive. At the
two exact centers both conditional entropies vanish, proving
\eqref{eq.isolated-instantaneous-minimum}.
\end{proof}

\begin{corollary}[Exact grid rate for an independent Laplace source]
\label{cor.iid-laplace-grid}
Let $(X_t)$ be independent with common density
\[
 f(x)=\frac1{2b}\exp(-|x|/b),\qquad b>0.
\]
Initialize MRR with any $\omega\geq1$ distinct representatives in
$\Lambda_{0,\eps}$ from \eqref{eq.lattice}, and with any distinct
timestamps. Set
\[
 r=\eps/b,\qquad q=e^{-r},\qquad \rho=e^{-2r},
 \qquad
 h_2(u)=-u\log_2u-(1-u)\log_2(1-u).
\]
Then its entropy rate is exactly
\begin{equation}\label{eq.iid-laplace-grid-rate}
 \mathcal H(r)=h_2(q)+q\left[
  1-\log_2(1-\rho)-\frac{\rho}{1-\rho}\log_2\rho
 \right],
 \qquad
 R_{\mathrm{MRR}}(X;A_0)=\mathcal H(r).
\end{equation}
The process $\eps$-entropy satisfies the certified bracket
\begin{equation}\label{eq.iid-laplace-floor-bracket}
 \max\left\{0,\log_2\frac e r\right\}
 \leq R_\eps(X)\leq\mathcal H(r).
\end{equation}
Thus \eqref{eq.iid-laplace-grid-rate} computes the achievable MRR
rate exactly; it does not assert equality with the process $\eps$-entropy.
For $r=10^{-3}$, outward rounding gives
\begin{equation}\label{eq.iid-laplace-grid-numeric}
 11.408479325551\ \leq\ R_\eps(X)\ \leq\ 11.408479565820
 \quad\text{bits per sample}.
\end{equation}
The width of the displayed certified interval is less than $2.403\times10^{-7}$ bits per sample.
\end{corollary}
\begin{proof}
Let $K_t=\floor{X_t/(2\eps)+1/2}$. Since the source has a density,
the lattice cell boundaries have probability zero. By
Proposition~\ref{prop.lattice-bijection},
$H(S_{1:n})=H(K_{1:n})=nH(K_1)$, where the last equality follows
from independence.

Integration over the lattice cells gives
\begin{equation}\label{eq.iid-laplace-grid-masses}
 \Prob(K_1=0)=1-q,\qquad
 \Prob(K_1=k)=\Prob(K_1=-k)
 =\frac q2(1-\rho)\rho^{k-1},\quad k\geq1.
\end{equation}
Conditional on $K_1\ne0$, its sign is uniform on $\{-1,1\}$ and
independent of its magnitude, whose law is geometric on
$\{1,2,\ldots\}$ with success probability $1-\rho$. The entropy
chain rule and the standard geometric-entropy formula
\cite{CoverThomas2006} therefore give
\[
 H(K_1)=h_2(q)+q\left[
  1-\log_2(1-\rho)-\frac{\rho}{1-\rho}\log_2\rho
 \right]<\infty.
\]
This proves \eqref{eq.iid-laplace-grid-rate}. The independent Laplace
source has differential entropy rate $\log_2(2eb)$. The Shannon
lower bound \eqref{eq.shannon-lower}, nonnegativity of entropy, and
the admissibility of the MRR reconstruction give
\eqref{eq.iid-laplace-floor-bracket}. Substitution of $r=10^{-3}$
gives the unrounded endpoints
\[
 \log_2(e/r)=11.40847932555105045097\ldots,
 \qquad
 \mathcal H(r)=11.40847956581995316836\ldots.
\]
Their difference is $2.40268902717\ldots\times10^{-7}$, which
proves the numerical statements.
\end{proof}

\begin{corollary}[Autoregressive Gaussian source]
\label{cor.ar-gaussian}
Let $X_t=aX_{t-1}+Z_t$, $|a|<1$, be stationary, where
the innovations are independent $N(0,\sigma^2)$ variables with
$\sigma>0$. Put
\[
 B_{\eps,\sigma}=\frac12\log_2
                         \frac{\pi e\sigma^2}{2\eps^2},
 \qquad
 G_{\eps,a,\sigma}=
             \frac{2(1+a^2)\eps^2}{\pi^2\sigma^2\ln2}.
\]
Then
\begin{equation}\label{eq.gaussian-floor-bracket}
 \max\{0,B_{\eps,\sigma}\}
 \leq R_\eps(X)\leq B_{\eps,\sigma}+G_{\eps,a,\sigma}.
\end{equation}
For every capacity and every fixed shared initial geometry contained
in a lattice of spacing $2\eps$,
\begin{equation}\label{eq.gaussian-lattice-rate}
 R_\eps(X)\leq R_{\mathrm{MRR}}(X;A_0)
            \leq B_{\eps,\sigma}+G_{\eps,a,\sigma}.
\end{equation}
For arbitrary fixed initialization, the upper rate bound becomes
\begin{equation}\label{eq.gaussian-general-rate}
 R_{\mathrm{MRR}}(X;A_0)
 \leq B_{\eps,\sigma}+G_{\eps,a,\sigma}
       +\frac{2\eps}     {e\sqrt{2\pi}\sigma\ln2}\left( \left\lfloor\frac{3\omega}{2}\right\rfloor-1\right).
\end{equation}
The case $a=0$ includes independent Gaussian sources, for which
$G_{\eps,0,\sigma}=2\eps^2/(\pi^2\sigma^2\ln2)$.
\end{corollary}
\begin{proof}
The stationary marginal law is
$N(0,\sigma^2/(1-a^2))$. Its Fisher information is
$J(\pi)=(1-a^2)/\sigma^2$, while the innovation Fisher information
is $J(f_Z)=1/\sigma^2$. The block factorization and score calculation
in the proof of Corollary~\ref{cor.ar-laplace} apply with
$\rho(z)=-z/\sigma^2$ and
$\rho_\pi(x)=-(1-a^2)x/\sigma^2$. The innovation scores have
mean zero, and the same independence argument removes every cross
term. Consequently,
\begin{equation}\label{eq.ar-gaussian-fisher}
 J_n=\frac{1-a^2+(n-1)(1+a^2)}{\sigma^2}
    =\frac{n+(n-2)a^2}{\sigma^2},
 \qquad
 \mathcal J=\frac{1+a^2}{\sigma^2}.
\end{equation}
All block laws are nondegenerate Gaussian, so the square-root
Sobolev and entropy integrability assumptions hold. Their finite
second moments imply $H(\floor{X_1})<\infty$. The conditional
density bound is $M=1/(\sqrt{2\pi}\sigma)$; it also bounds the
first marginal density, since that marginal has variance at least
$\sigma^2$. Finally,
\[
 \overline h(X)=h(Z_1)=\frac12\log_2(2\pi e\sigma^2).
\]
The bounded-density refinement of Theorem~\ref{thm.near-floor},
 with $M=1/(\sqrt{2\pi}\sigma)$, gives
 \eqref{eq.gaussian-general-rate}. The invariant lattice family in
Proposition~\ref{prop.lattice} removes the volume term, giving
\eqref{eq.gaussian-lattice-rate}. Its admissibility and the Shannon
lower bound give \eqref{eq.gaussian-floor-bracket}.
\end{proof}

\subsection{I.i.d. inputs and invariant state laws}
\label{sec.iid-state-rates}
For i.i.d. inputs, the address history contains no predictive
information about the next input beyond that needed to determine the
current partition. This yields an exact state-conditional expression for
the block entropy, even when the state process is not stationary.

\begin{lemma}[Independent-source rate identity]\label{lem.iid-rate}
Let the inputs be independent with common law $\mu$, and initialize MRR
at a fixed shared state $A_0$. Write $a_t$ for its age-ranked state,
identify $\val(a)$ with the set of coordinates of an age-ranked tuple,
and define
\begin{equation}\label{eq.state-entropy-observable}
 p_s(a)=\mu\{x:S_{\val(a)}(x)=s\},\qquad
 g_\mu(a)=-\sum_{s\in\mathcal S}p_s(a)\log_2p_s(a),
\end{equation}
with $0\log_2 0=0$. Suppose $H(\floor{X_1})<\infty$. Then all the
following entropies are finite, and for every $n\geq1$,
\begin{equation}\label{eq.iid-state-entropy-identity}
 H(S_{1:n})
 =\sum_{t=1}^n H(S_t\mid A_{t-1})
 =\sum_{t=1}^n\E g_\mu(a_{t-1}).
\end{equation}
If $\eta_t=\mathcal L(a_t)$ and
$\bar\eta_n=n^{-1}\sum_{t=0}^{n-1}\eta_t$, then
\begin{equation}\label{eq.iid-occupation-rate}
 \frac{H(S_{1:n})}{n}=\int g_\mu\,d\bar\eta_n,
 \qquad
 R_{\mathrm{MRR}}(X;A_0)
 =\limsup_{n\to\infty}\int g_\mu\,d\bar\eta_n.
\end{equation}
\end{lemma}
\begin{proof}
Decoder compatibility and the fixed initialization make $A_{t-1}$,
and hence $a_{t-1}$, measurable functions of $S_{1:t-1}$. The address
history is a function of $X_{1:t-1}$, so it is independent of $X_t$.
Consequently, for every address $s$,
\[
 \Prob(S_t=s\mid S_{1:t-1})
 =p_s(a_{t-1})
 =\Prob(S_t=s\mid A_{t-1})
 \quad\text{almost surely}.
\]
The conditional entropies given the history and given the current state
are therefore equal, and both are $\E g_\mu(a_{t-1})$.

For completeness, finiteness follows from a uniform bound that does not
require bounded support. Every unit interval meets at most
$K_{\omega,\eps}=2\omega+2+\ceil{1/(2\eps)}$ address cells:
there are at most $m_\omega$ short cells in the entire partition,
and the remaining cells have disjoint interiors of length $2\eps$.
Thus, at every fixed state,
\[
 0\leq g_\mu(a)
 \leq H(\floor{X_1})+\log_2K_{\omega,\eps}<\infty.
\]
The entropy chain rule now proves
\eqref{eq.iid-state-entropy-identity}. Integrating with respect to the
state laws gives \eqref{eq.iid-occupation-rate}.
\end{proof}
The state-conditional entropy identity requires independent
inputs: for dependent sources, the current state may discard
predictive information contained in the address history.
\begin{proposition}[Causal floor for independent inputs]\label{prop.causal-iid}
Let the inputs be independent with common law $\mu$. Then
\[
R^{\mathrm{c}}_\varepsilon(X)=F_\varepsilon(\mu).
\]
Consequently $F_\varepsilon(\mu)\le R_{\mathrm{MRR}}(X;A_0)$ for every fixed
shared initial state. If $A_0$ is $\varepsilon$-separated, then
$R_{\mathrm{MRR}}(X;A_0)\le F_\varepsilon(\mu)+2$.
\end{proposition}
 
\begin{proof}
For the upper bound, apply an admissible scalar partition $Q$ in
\eqref{eq.scalar-floor} to each input separately. The labels are independent with
entropy $H_\mu(Q(X))$ per sample. Taking the infimum over $Q$ gives
$R^{\mathrm{c}}_\varepsilon(X)\le F_\varepsilon(\mu)$, trivially if $F_\eps(\mu)=\infty$.
 
For the lower bound, fix a causal code. Since $L_{1:t-1}$ is a function of
$X_{1:t-1}$, conditioning on $X_{1:t-1}$ cannot increase entropy, and
\[
  H(L_{1:n})=\sum_{t=1}^n H(L_t\mid L_{1:t-1})
  \ge\sum_{t=1}^n H(L_t\mid X_{1:t-1}).
\]
Fix $x_{1:t-1}$. The map $u\mapsto g_t(x_{1:t-1},u)$ partitions $\mathbb{R}$
into countably many Borel cells. On the cell with label $\ell$, the value
$\phi_t\bigl(g_1(x_1),\dots,g_{t-1}(x_{1:t-1}),\ell\bigr)$ lies within
$\varepsilon$ of every point, so the partition is admissible in
\eqref{eq.scalar-floor}. Since $X_t$ is independent of $X_{1:t-1}$, the conditional
law of $L_t$ given $X_{1:t-1}=x_{1:t-1}$ is the image of $\mu$ under this
partition, and its entropy is at least $F_\varepsilon(\mu)$. Averaging over
$x_{1:t-1}$ and summing over $t$ gives $H(L_{1:n})\ge nF_\varepsilon(\mu)$,
hence $R^{\mathrm{c}}_\varepsilon(X)\ge F_\varepsilon(\mu)$.
 
The lower bound for MRR follows from~\eqref{eq.causal-sandwich}. For the upper
bound we may assume $F_\varepsilon(\mu)<\infty$. A faithful cell meets at most
$\lceil 2\varepsilon\rceil+1$ unit intervals, so~\eqref{eq.partition-overlap} gives
$H(\lfloor X_1\rfloor)<\infty$ and Lemma~\ref{lem.iid-rate} applies. By
Proposition~\ref{prop.separation}, every reachable geometry is
$\varepsilon$-separated. Theorem~\ref{thm.two-bit} therefore gives
$g_\mu(a)\le F_\varepsilon(\mu)+2$ at every reachable age-ranked state, and
Lemma~\ref{lem.iid-rate} gives
$H(S_{1:n})\le n\bigl(F_\varepsilon(\mu)+2\bigr)$.
\end{proof}
 
With lattice initialization and consistent boundary ownership,
Proposition~\ref{prop.lattice-bijection}
and~\eqref{eq.lattice-scalar-one-bit} improve the upper
bound to $F_\varepsilon(\mu)+1$.  
Under the pathwise criterion, therefore, the causal floor of an independent
source is its scalar $\varepsilon$-entropy, and MRR is within two bits of it from any separated
initialization.
Neuhoff and
Gilbert~\cite{NeuhoffGilbert1982} need time-sharing between two scalar
quantizers because average distortion can be traded against rate; a pathwise
constraint admits no such trade-off.

\begin{proposition}[Rates from limiting state occupation laws]
\label{prop.invariant-state-rate}
Under the hypotheses of Theorem~\ref{thm.feller}, take $\mu$ to be the
common input law and fix $a_0\in\mathcal K_{J,\delta}$. The function
$g_\mu$ in \eqref{eq.state-entropy-observable} is bounded and continuous
on $\mathcal K_{J,\delta}$. Every weak subsequential limit $\pi$ of
$\bar\eta_n$ is invariant for the age-ranked transition kernel $P$.
If $\bar\eta_{n_j}\Rightarrow\pi$, then
\begin{equation}\label{eq.subsequential-invariant-rate}
 \lim_{j\to\infty}\frac{H(S_{1:n_j})}{n_j}
 =\int g_\mu(a)\,\pi(da).
\end{equation}
In particular, if $\bar\eta_n\Rightarrow\pi$, the normalized block
entropy has a limit and
\begin{equation}\label{eq.invariant-state-rate}
 R_{\mathrm{MRR}}(X;A_0)
 =\lim_{n\to\infty}\frac{H(S_{1:n})}{n}
 =\int g_\mu(a)\,\pi(da).
\end{equation}
Convergence of $\bar\eta_n$ holds, in particular, if $P$ has a unique
invariant probability measure on $\mathcal K_{J,\delta}$.
\end{proposition}
\begin{proof}
The input interval $[L_0,U_0]$ and the state interval $J$ are bounded.
The formulas for the refinement indices therefore provide a common
finite set $\mathcal S_0\subset\mathcal S$ containing every address
that can be emitted on this input interval at any state in
$\mathcal K_{J,\delta}$. In particular,
$0\leq g_\mu(a)\leq\log_2|\mathcal S_0|$ uniformly in $a$.

Suppose $a^{(m)}\to a$ in $\mathcal K_{J,\delta}$. The positive
separation $\delta$ makes the spatial ordering of the coordinates
eventually constant. As in the proof of Theorem~\ref{thm.feller},
exclude the finite set of halo endpoints, nearest-neighbor boundaries,
and refinement cuts for the limiting state. Away from this set, every
comparison selecting an address is eventually unchanged; a refinement
cell that appears when a gap crosses a multiple of $2\eps$ has
vanishing length at its endpoint. It follows that, for every
$s\in\mathcal S_0$,
\[
 \1_{\{S_{\val(a^{(m)})}(x)=s\}}
 \longrightarrow\1_{\{S_{\val(a)}(x)=s\}}
 \quad\text{for }\mu\text{-almost every }x.
\]
Here atomlessness makes the excluded finite set null. Dominated
convergence gives $p_s(a^{(m)})\to p_s(a)$ for every $s\in\mathcal S_0$.
Continuity of $-u\log_2u$ on $[0,1]$, with its value at zero defined
as zero, proves continuity of $g_\mu$.

Compactness ensures that the sequence of probability measures
$\bar\eta_n$ has weakly convergent subsequences. The Markov evolution
gives the identity
\[
 \bar\eta_nP-\bar\eta_n=\frac{\eta_n-\eta_0}{n}.
\]
Let $\bar\eta_{n_j}\Rightarrow\pi$ and let $\varphi$ be continuous
on $\mathcal K_{J,\delta}$. The function $P\varphi$ is continuous
by Theorem~\ref{thm.feller}; both functions are bounded by compactness.
Testing the preceding identity against $\varphi$ and passing to the
limit yields $\int P\varphi\,d\pi=\int\varphi\,d\pi$.
Thus $\pi$ is invariant. Applying weak convergence to the bounded
continuous function $g_\mu$ and using
\eqref{eq.iid-occupation-rate} proves
\eqref{eq.subsequential-invariant-rate} and
\eqref{eq.invariant-state-rate}. If the invariant probability measure
is unique, every convergent subsequence of $\bar\eta_n$ has the same
limit. Compactness then implies convergence of the entire sequence.
\end{proof}

\section{Discussion}\label{sec.discussion}
The representative-state framework separates the geometry of reconstruction,
the information needed to update that geometry, and the coding of the
resulting address sequence. The midpoint realization has a deterministic
error guarantee for every input sequence and every positive capacity.
Its scalar entropy comparison in Theorem~\ref{thm.two-bit} is uniform over
source distributions, subject to separation of the representative geometry.
The same geometric bound controls the branching of the causal address
tree and gives a distribution-free process comparison in
Corollary~\ref{cor.two-bit-process}. Fisher information of block
densities gives the tolerance-dependent improvement in
Theorem~\ref{thm.near-floor}. Both comparisons are against arbitrary
admissible block partitions.

The two terms in the process bound identify different effects. The density
term measures variation of the source density within a reconstruction
rectangle. The volume term accounts for the loss caused by short proximity and
remainder cells. 
Finite Fisher information controls both terms; a uniform bound on
 the one-step conditional densities is needed only for the optional
 refinement of the volume estimate.
Lattice alignment removes the latter term and gives the
sharpest AR-Laplace process $\eps$-entropy certificate used here. This holds for every
grid origin and every positive spacing $2\eps$, with no power-of-two
restriction on resolution. The grid conclusions highlighted in
Section~\ref{sec.grid-certificates} give a sharp one-bit scalar comparison
and the process bound \eqref{eq.lattice-combined-gap}, uniformly over
capacity and grid origin. For a fixed source satisfying the stated
Fisher-information and entropy assumptions, the rate gap therefore
vanishes as $\eps\downarrow0$ even when capacity grows without restriction.
In the
invariant lattice family, however, consistent boundary ownership makes
the numerical reconstruction identical to a fixed uniform quantizer.
The stored subset and local addresses still evolve, but their evolution
does not change block entropy or entropy rate. The one-bit lattice
comparison follows from the entropy chain rule and a cell-overlap count;
the causal address correspondence explains its applicability at every
representative capacity. Corollary~\ref{cor.capacity-one-grid} makes
this equivalence unavoidable at capacity one: every initial representative
anchors a fixed grid. The entropy
identity in Proposition~\ref{prop.rectangle-identity} applies to both cases
and can be used to study alternative refinement and replacement rules.

Non-lattice geometries can also have zero expected volume loss when the
source is supported on full-width proximity cells.
Theorem~\ref{thm.exact-separated} identifies the exact block floor for
uniform noise around separated centers and gives explicit sufficient
initializations for MRR to attain it, without a common-lattice assumption
on the centers. In
Proposition~\ref{prop.capacity-advantage}, storing two non-lattice centers
attains one bit per sample, whereas every fixed grid of width $2\eps$
has rate at least $3/2$ bits per sample, even with the best offset and
optimal block entropy coding. This demonstrates a benefit of representative
geometry under source-matched initialization. The attaining state is
absorbing, so the comparison does not establish a benefit of FIFO adaptation
or convergence from an arbitrary state. In this two-band case, Remark~\ref{rem.isolated-geometry}
shows that this unordered absorbing geometry is unique and is a strict
local minimizer of the one-step address entropy; it is not an isolated
point of the continuous state space. The capacity-dependent universal
upper bounds likewise do not imply monotonicity of the actual rate. 
More generally,
Corollary~\ref{cor:sharp-uniform-band-penalty} identifies
$1/(2\ln 2)$ bits per sample as the sharp worst-case excess
rate of the best uniform grid over the separated uniform-band
source class.

For i.i.d. inputs, Lemma~\ref{lem.iid-rate} gives a direct link
between the representative dynamics and the entropy rate. The normalized
block entropy is exactly the integral of the state-conditional address
entropy against the averaged state law. Under the compact Feller
hypotheses, Proposition~\ref{prop.invariant-state-rate} identifies
subsequential limits with integrals against invariant occupation laws.
Convergence of these averaged laws, for example under uniqueness of the
invariant measure, yields a rate limit. Existence of an invariant measure
alone does not imply this convergence. For dependent inputs, the current
state can discard predictive information available in the address history,
so the same state-conditional identity need not hold.

The capacity constraint counts stored representatives. It does not bound
the number of bits needed to represent arbitrary real values, the size of
timestamps, or the memory of an entropy coder. Block coding of addresses
can exploit dependence beyond the representative state. The statements
about $R_\eps$ concern this entropy rate; a coding implementation with a
specified memory and delay requires its own analysis.

The Fisher-information estimate applies to unbounded Gaussian scores and
retains a finite-tolerance guarantee. The asymptotic Shannon lower-bound
results of Linder--Zamir and Koch concern expected distortion
\cite{LinderZamir1994,Koch2016}; the bounds here concern a
coordinatewise error constraint along each path. Further questions include density
conditions beyond finite Fisher information, quantitative convergence of the
age-ranked state, and replacement rules that control the expected
cell-volume loss directly. These questions can be posed without changing
the four structural properties in Section~\ref{sec.axioms}.
\newcommand{\etalchar}[1]{$^{#1}$}
\providecommand{\bysame}{\leavevmode\hbox to3em{\hrulefill}\thinspace}
\providecommand{\MR}{\relax\ifhmode\unskip\space\fi MR }
\providecommand{\MRhref}[2]{%
  \href{http://www.ams.org/mathscinet-getitem?mr=#1}{#2}
}
\providecommand{\href}[2]{#2}

\appendix
\section{Certified evaluation of the   process $\eps$-entropy}\label{app.computation}
Let $X$ be independent and identically distributed on
$x_1<\cdots<x_N$, with probabilities $p_i>0$. Form the graph with
vertex set $\{1,\ldots,N\}$ and an edge between $i$ and $j$ precisely
when $|x_i-x_j|>2\eps$. An independent set is exactly a set of source
values admitting a common $\eps$-accurate reconstruction. Let
$C_1,\ldots,C_m$ be its maximal independent sets and put
\begin{equation}\label{eq.incidence}
 A_{ij}=\1_{\{i\in C_j\}}.
\end{equation}
Every maximal set is a consecutive index interval. Indeed, inserting any
omitted point between its extreme values preserves its diameter. Its
right endpoint is the largest admissible index for its left endpoint;
therefore $m\leq N$. Sorting followed by a moving right endpoint
enumerates all candidates, after which contained candidates are removed.
This gives an interval-incidence representation of size at most $N^2$.

The asymptotic chromatic-entropy characterization of graph entropy
\cite{AlonOrlitsky1996} and its convex formulation
\cite{Csiszar1990} give
\begin{equation}\label{eq.graph-floor}
 R_\eps(X)=\min_{q\in\Delta_m}U(q),\qquad
 U(q)=-\sum_{i=1}^N p_i\log_2 b_i(q),\qquad
 b_i(q)=\sum_{j=1}^m A_{ij}q_j,
\end{equation}
where $\Delta_m=\{q\in[0,1]^m:\sum_jq_j=1\}$. The graph used here is
the incompatibility graph, so independent sets, rather than cliques,
are the admissible reconstruction sets. A block cell is independent
in the disjunctive graph power exactly when its source vectors have
coordinatewise diameter at most $2\eps$. This identifies the
partition benchmark in \eqref{eq.block-floor} with the chromatic
entropy to which the cited asymptotic characterization applies.

For any feasible $q$ with $b_i(q)>0$, define
\begin{equation}\label{eq.graph-certificate}
 c(q)=\max_{1\leq j\leq m}
             \sum_{i=1}^N\frac{p_iA_{ij}}{b_i(q)}.
\end{equation}
The dual certificate of \cite[Theorem~27]{Harangi2023}, specialized to
the unconditional case, is
\begin{equation}\label{eq.certified-floor}
 U(q)-\log_2c(q)\leq R_\eps(X)\leq U(q).
\end{equation}
Here $c(q)\geq1$, since the $q$-weighted average of the quantities
inside the maximum is one. Equality $c(q)=1$ is equivalent to
optimality.

One implementation uses their alternating minimization iteration,
which in this interval-incidence form is
\begin{equation}\label{eq.graph-iteration}
 q_j^{(k+1)}=q_j^{(k)}
     \sum_{i=1}^N\frac{p_iA_{ij}}{b_i(q^{(k)})},
 \qquad q_j^{(0)}>0.
\end{equation}
The entries remain nonnegative and sum to one. Convergence of the
objective to \eqref{eq.graph-floor} follows from
\cite[Theorem~3]{Harangi2023}. At any iteration,
$\log_2c(q^{(k)})$ is a certified optimization error. Rigorous numerical
certification requires outward rounding when evaluating the two
endpoints in \eqref{eq.certified-floor}; an unrounded optimizer output
alone does not provide such a certificate.

For an independent source uniform on $\{0,1,2\}$ and
$\eps=0.6$, the maximal admissible sets are $\{0,1\}$ and $\{1,2\}$.
Taking $q=(1/2,1/2)$ gives $b=(1/2,1,1/2)$,
$U(q)=2/3$, and $c(q)=1$. Thus
\begin{equation}\label{eq.discrete-example}
 R_{0.6}(X)=\frac23,
 \qquad F_{0.6}(\mathcal L(X))=h_2(1/3)>\frac23,
\end{equation}
where $h_2(u)=-u\log_2u-(1-u)\log_2(1-u)$.
The scalar value follows because no cell can contain both $0$ and $2$;
a scalar partition must separate them, and the minimum is attained by
joining $1$ to either endpoint. This example distinguishes the scalar
and process $\eps$-entropies for the same independent source.
By Proposition~\ref{prop.causal-iid}, $F_{0.6}(\mathcal{L}(X))$ is also the
optimal causal rate for this source. The example therefore exhibits a strictly
positive price of causality at a fixed tolerance:
$h_2(1/3)-2/3=\log_2 3-4/3\approx0.252$ bit per sample.

\end{document}